\documentclass[sigconf, authorversion, nonacm]{acmart}
\acmConference[Correctness 2026]{10th International Workshop on
  Software Correctness for HPC Applications}{November 16,
  2026}{Chicago, IL}

\usepackage{xcolor,stmaryrd,graphicx,siunitx,doi,tagging}
\usepackage{listings}
\lstdefinestyle{arxivC}{
    backgroundcolor=\color{gray!5},   
    keywordstyle=\color{blue},
    stringstyle=\color{purple},
    commentstyle=\color{green!60!black},
    basicstyle=\ttfamily\small,
    breaklines=true,                 
    numbers=left,
    numbersep=8pt,     % Distance between line numbers and code
    xleftmargin=40pt   % Indents the whole listing so line numbers sit inside the margin
}

\graphicspath{{images/}}
\newtheorem{theorem}{Theorem}

\DeclareFontFamily{U}{cmtex}{}
\DeclareFontShape{U}{cmtex}{m}{n}{
 <-> cmtex10
}{}

\newcommand{\mysf}[1]{\textup{\textsf{#1}}}  % upright sans-serif
\newcommand{\code}[1]{\textup{\texttt{#1}}}
\newcommand{\pcstyle}[1]{\textcolor{black}{\tt #1}}

\newcommand{\lb}{\code{\char`\{}}
\newcommand{\rb}{\code{\char`\}}}
\newcommand{\U}{{\ttfamily\symbol{'137}}}
\newcommand{\lp}{\code{(}}
\newcommand{\rp}{\code{)}}

\newcommand{\bs}{\code{\char`\\}}

\newcommand{\ra}{\rightarrow}

\newcommand{\mplite}{\textup{\textsc{MPlite}}}

\newcommand{\gtxt}[1]{\ensuremath{\mathbf{\color{red}{#1}}}}
\newcommand{\gsep}{\gtxt{\hspace{.333ex}|\hspace{.333ex}}}
\newcommand{\gas}{\gtxt{::=}}

\newcommand{\Expr}{\mysf{Expr}} % set o expressions
\newcommand{\T}{\textit{true}}
\newcommand{\F}{\textit{false}}
\newcommand{\PID}{\mysf{PID}}

\newcommand{\psend}{\pcstyle{send}}
\newcommand{\precv}{\pcstyle{recv}}
\newcommand{\pto}{\pcstyle{to}}
\newcommand{\pfrom}{\pcstyle{from}}

\newcommand{\mpiinit}{\code{MPI\U{}Init}}
\newcommand{\mpifinalize}{\code{MPI\U{}Finalize}}
\newcommand{\mpicommworld}{\code{MPI\U{}COMM\U{}WORLD}}

\newcommand{\mpisend}{\code{MPI\U{}Send}}
\newcommand{\mpirecv}{\code{MPI\U{}Recv}}
\newcommand{\mpianytag}{\code{MPI\U{}ANY\U{}TAG}}
\newcommand{\mpisendrecv}{\code{MPI\U{}Sendrecv}}
\newcommand{\mpianysource}{\code{MPI\U{}ANY\U{}SOURCE}}
\newcommand{\mpibcast}{\code{MPI\U{}Bcast}}
\newcommand{\mpigather}{\code{MPI\U{}Gather}}
\newcommand{\mpiisend}{\code{MPI\U{}Isend}}
\newcommand{\mpiirecv}{\code{MPI\U{}Irecv}}
\newcommand{\mpiwait}{\code{MPI\U{}Wait}}

\newcommand{\kwsty}[1]{\textup{\textcolor{black}{\texttt{#1}}}}
\newcommand{\kwparam}{\kwsty{param}}
\newcommand{\kwvar}{\kwsty{var}}

\newcommand{\kwwhile}{\kwsty{while}}

\newcommand{\kwchoose}{\kwsty{choose}}
\newcommand{\kwassert}{\kwsty{assert}}
\newcommand{\kwassume}{\kwsty{assume}}
\newcommand{\kwsend}{\kwsty{send}}
\newcommand{\kwto}{\kwsty{to}}
\newcommand{\kwrecv}{\kwsty{recv}}
\newcommand{\kwfrom}{\kwsty{from}}
\newcommand{\semi}{\kwsty{;}}
\newcommand{\comma}{\kwsty{,}}
\newcommand{\coleq}{\kwsty{:=}}
\newcommand{\eq}{\kwsty{=}}
\newcommand{\st}{\kwsty{|}}
\newcommand{\bbar}{\kwsty{||}}
\newcommand{\ident}[1]{\textit{#1}}
\newcommand{\fun}[1]{\textsf{#1}}
\newcommand{\vpid}{\fun{pid}}
\newcommand{\vNP}{\fun{NP}}

\newcommand{\nummsg}{\fun{nummsg}}
\newcommand{\msginv}{\fun{msginv}}
\newcommand{\slevel}{\fun{slevel}}
\newcommand{\rlevel}{\fun{rlevel}}
\newcommand{\HB}{\textsf{HB}}

\newcommand{\anummsg}{\code{nummsg}}
\newcommand{\amsginv}{\code{msginv}}
\newcommand{\aslevel}{\code{slevel}}

\newcommand{\nregions}{\code{nregions}}
\newcommand{\region}{\code{region}}
\newcommand{\vmnp}{\code{VM\U{}NP}}

\newcommand{\vmlevel}{\code{VM\U{}level}}

\newcommand{\vmfc}{VMFC} % name of tool
\newcommand{\Wp}{\textsf{Wp}} % the Frama-C Plug-in

\newcounter{lc}  % global variable for line count
\DeclareRobustCommand{\showline}{%
  \stepcounter{lc}%
  {\makebox[1mm][r]{\textcolor{blue}{\scriptsize\arabic{lc}}}\hspace{1.5mm}}%
}

\newenvironment{program}%
  {%
    \setcounter{lc}{0}%
    \begin{math}%
      \begin{array}[t]{@{\showline}l@{}}%
  }%
  {%
      \end{array}%
    \end{math}%
  }

\newcommand{\proofloc}{Appendix \ref{sec:proof}}

\begin{document}

\title{Parameterized Verification of Deterministic MPI Programs}

\author{Stephen F.\ Siegel}
\email{siegel@udel.edu}
\orcid{0000-0001-9359-3332}
\affiliation{%
  \institution{University of Delaware}
  \city{Newark}
  \state{Delaware}
  \country{USA}
}

\begin{abstract}
  We consider the problem of verifying a message passing program in
  which the number of processes is a parameter \mysf{NP} and each
  process knows its unique ID.  Processes communicate using send and
  receive commands which specify a single destination or source.  To
  verify the program, the user provides functions specifying the
  number of messages sent from process $i$ to process $j$, the
  \emph{level} of each communication event in the happens-before
  hierarchy, and a fact that holds for the $k^\textit{th}$ message
  sent from $i$ to $j$.  These are used to transform the program to a
  parameterized sequential program which can be verified using any
  techniques appropriate for such programs.  We realize this approach
  in an extension to Frama-C/\Wp{} to verify C/MPI programs.
\end{abstract}


\maketitle 

\section{Introduction}
\label{sec:intro}

We consider message-passing programs in which the number of processes,
\vNP, occurs as a parameter.  The program is executed by specifying a
concrete value for \vNP.  All $\vNP$ processes then execute the program
code, interacting only by sending and receiving messages.  A process
can read \vNP{} as well as its \vpid, a unique process ID number.
The goal of \emph{parameterized verification} is to verify correctness
properties, such as absence of deadlock or assertion violations, for
all $\vNP\geq 1$.

The Message Passing Interface (MPI) \cite{mpi-forum:2015:mpi31} is a
widely-used API for writing such programs, especially in high
performance computing.  Despite its many bells and whistles, the basic
MPI model is the one described above.  Many important algorithms
can be expressed using only a handful of MPI operations, such as
\mpisend, \mpirecv, and \mpisendrecv.

There is growing interest in verification in the scientific computing
community (e.g., \cite{nsf:2026:cs2}), especially with the advent of
AI-generated code.  Several projects have focused on verifying MPI
programs.  Model checking, for some finite number of values of $\vNP$,
is one popular approach.  But model checking suffers from the state
explosion problem: the blowup in the number of states as $\vNP$
increases.  Recourse to the ``small scope hypothesis'' leaves many
practitioners unconvinced, as their production runs may involve
thousands of processes, and defects often manifest only at large and
seemingly random process counts.

Significant progress has been made on verifying \emph{deterministic}
message-passing programs.  In these programs, each receive
operation specifies a single source process from which it will accept
a message.  In MPI, this prohibits the use of \mpianysource{}.  Many
(probably most) MPI programs can be written using a deterministic
subset, and programmers are strongly encouraged to avoid
nondeterministic constructs unless absolutely necessary.

Deterministic programs are easier to verify.  In particular, given any
input, one need only examine a single interleaving to determine if a
program deadlocks or violates an assertion.  In theory, this could
solve the state explosion problem.  In practice, performance of model
checking tools still degrades for other reasons (such as branch
explosion from the input space), and does not come close to the scale
used in practice.

There has been recent progress on parameterized verification of
deterministic message-passing programs.  Notably, an approach based on
the deductive program verification tool Dafny \cite{dafny-web,
  leino:2010:dafny} is described in \cite{fedchin-etal:2026:dafnympi}.
Following this approach, the user implements an algorithm in an
extension of the Dafny language which supports MPI.  The
implementation should specify an integer ID for each message and use
this ID as the tag in point-to-point operations.  The user provides
additional information such as functions from tags to the source,
destination, and message content; functions specifying occurrences of
collective function calls; a Dafny specification of the expected
output of the program; and additional information, such as loop
invariants, enabling the verifier to prove that the program is
deadlock-free and produces the expected result.  The program can be
proved correct for any number of processes, and for any inputs.  The
tool can also generate a Python implementation that uses the mpi4py
and NumPy packages to implement MPI and array operations,
respectively.

In this paper, we advance this line of research by simplifying and
generalizing in several directions.  Intead of generating a Python
implementation, our approach consumes a C/MPI program, without regard
to how that program was produced.  The user is not required to use
tags in any particular way, nor is the program modified for
verification.  Instead, the user inserts annotations (structured
comments) into the code to provide the additional information needed
to verify the program.  Using this information, the program is
automatically transformed to a sequential C program using three
operations that are commonly supported by C verification tools:
nondeterministic choice, assertions, and assumptions. The sequential C
program could be verified using any appropriate verification tool for
sequential C.

Here, we explore deductive verification using Frama-C
\cite{cuoq-etal:2012:frama-c, kosmatov-etal:2024:frama-c} and its
\Wp{} plug-in \cite{baudin-etal:2025:wp-manual,
  blanchard-etal:2024:wp}.  We have extended Frama-C's ACSL
specification language to enable parameterized verification of C/MPI
programs that are restricted to the deterministic subset of MPI
consisting of \mpisend, \mpirecv, and \mpisendrecv, but \emph{not}
\mpianysource.  The approach is modular: procedures can be verified
individually, and the body of a procedure may be broken down into
distinct communication-closed regions, each with its own
specification.  A prototype tool transforms the C/MPI program with
extended ACSL specifications into a sequential C program with ordinary
ACSL specifications.  This sequential program is then analyzed by
Frama-C/\Wp{}, which generates verification conditions that are
discharged by automated theorem provers.

We have applied this technique to yield complete proofs of five C/MPI
programs.  These programs are verified to be free of deadlocks and
assertion violations, to terminate on all inputs, and to compute the
expected results, without any bounds (other than those imposed by C's
type system) on the number of processes, number of discrete points, or
other parameters.

\section{Verification Approach}
\label{sec:transform}

In this section we describe the basic approach using a toy,
untyped message-passing language we call \mplite.

\subsection{A simple message-passing language}

Fig.\  \ref{fig:mplite} presents an abstract grammar for \mplite.
Values include integers, reals, and arrays.  Expressions are total,
deterministic, side-effect-free functions of their arguments.  Rather
than specify an expression grammar, we will use operators whose
meanings are, hopefully, obvious.  Where a boolean condition is
required, $0$ is interpreted as \F{} and any other value as \T.

\begin{figure}
  \centering
  \begin{tabular}{@{}r@{\hspace{6pt}}c@{\hspace{6pt}}l@{}}
    \textit{program}
    & \gas & \gtxt{(}\kwparam\ $\overline{x}$\code{;}\gtxt{)?}
             \gtxt{(}\kwvar\ $\overline{x}$\code{;}\gtxt{)?}
             \gtxt{(}\kwassume\ $e$\code{;}\gtxt{)?}
             $s$\gtxt{*}\\
    $\overline{x}\,\gtxt{\in}\,\textit{idlst}$
    & \gas &  $x$ \gtxt{(}\code{,} $x$\gtxt{)*} \\
    $s\,\gtxt{\in}\,\textit{stmt}$
    & \gas & $x$\code{:=}$e$\code{;}
             \gsep{} $x$\code{[}$e$\code{]}\code{:=}$e$\code{;}
             \gsep{} \code{if} \code{(}$e$\code{)} $s$
             \gtxt{(}\code{else} $s$\gtxt{)?}\\
    \multicolumn{3}{l}{
    \hspace{3ex} \gsep{} \code{while} \code{(}$e$\code{)} $s$
    \gsep{} \lb\ $s$\gtxt{*}\ \rb\
    \gsep{} \code{choose} $x$\ \st\ $e$\code{;}
    }\\
    \multicolumn{3}{l}{
    \hspace{3ex} \gsep{} \code{assert} $e$\code{;}
    \gsep{} \code{send} $e$ \code{to} $e$\code{;}
    \gsep{} \code{recv} $x$ \code{from} $e$\code{;}
    }
  \end{tabular}
  \begin{center}
    $x$ $\gtxt{\in}$ $\textit{ID}$
    \hspace{2em}
    $e$ $\gtxt{\in}$ $\Expr$
  \end{center}
  \caption{\mplite{} syntax. $\textit{ID}$ is a set of identifiers.
    $\Expr$ is a set of deterministic side-effect-free expressions.
    The reserved parameter $\vNP$ evaluates to the number of processes,
    and the reserved read-only variable $\vpid$ returns the unique ID
    of this process.}
  \label{fig:mplite}
\end{figure}

A program begins with optional parameter and variable declarations.
Parameters are read-only and \emph{universal}: they have the same
value on every process.  The reserved parameter $\vNP$ holds the
number of processes executing the program.  Variables are
process-local: each process has its own store assigning values to
variables.  Processes have unique IDs in $\PID=\{0,\ldots,\vNP-1\}$;
the reserved read-only variable \vpid{} holds a process's ID.  The
declarations are followed by an optional assumption, which places
restrictions on the possible parameter values and initial variable
values.

Statements include assignment (to a variable or array element),
\emph{if-else}, \emph{while} loops, and compound statements.
Statement $\kwchoose\ x\ \st\ e\semi$ assigns some value to $x$ that
satisfies $e$.  If there is no such value, execution halts.

The command ``\kwsend\ \ident{data} \pto\ \ident{dest}'' sends the
value of \ident{data} to the process with ID \ident{dest}.  There is
one FIFO channel for each ordered pair of processes $s\ra d$ and the
effect of \kwsend{} is to enqueue the message on the channel for which
$s$ is the ID of the sender and $d$ is \textit{dest}.  The channels
are unbounded, so \psend{} never blocks.  Command ``\precv\
\textit{buf}\ \pfrom\ \textit{source}'' removes the oldest buffered
message originating from \textit{source} and stores it in variable
\textit{buf}; this command blocks until a message becomes available.
A \textit{dest} or \textit{source} not in $\PID$ results in a
no-op.

A \emph{process state} includes a program counter and assigns a value
to each variable.  A (global) \emph{state} comprises a process state
for each process, an assignment of values to the parameters (which do
not change during an execution), and a finite sequence of values for
each $(s,d)\in\PID\times\PID$ specifying the state of channel
$s\ra d$.  In an \emph{initial state}, every process is at the program
start location, all channels are empty, and the values assigned to
parameters and variables cause the assumption (if present) to evaluate
to $\T$.

An \emph{execution prefix} is a finite or infinite sequence of
alternating states and transitions, starting from an initial state.
Each transition is an atomic step taken by one process and updates the
state as described above.  An \emph{execution} is a maximal execution
prefix.
% , i.e., one which cannot be extended to a longer execution prefix.

An execution prefix is \emph{correct} if, for all statements of the
form \code{assert} $\phi$ occurring on the prefix, $\phi$ evaluates to
$\T$.  A program is \emph{correct} if all of its executions are
correct, and in each execution, all processes terminate normally,
i.e., by completing execution of the last statement.

\newcommand{\vsum}{\ident{sum}}
\newcommand{\vleft}{\ident{left}}
\newcommand{\vright}{\ident{right}}
\newcommand{\visRealArray}{\fun{isRealArray}}

\begin{figure}
  \hspace{3ex}
  \begin{program}
    \kwparam\ A\semi\\
    \kwvar\ i\comma\ x\comma\ \vsum\comma\ \vleft\comma\ \vright\semi\\
    \kwassume\ \visRealArray(A,\vNP)\ \wedge\ x=A[\vpid]\semi\\
    \vleft\ \coleq\ (\vpid+\vNP-1)\%\vNP\semi\\
    \vright\ \coleq\ (\vpid+1)\%\vNP\semi\\
    i\ \coleq\ 1\semi\\
    \vsum\ \coleq\ x\semi\\
    \kwwhile\ \lp\ i<\vNP\ \rp\ \lb\\
    \ \ \kwsend\ x\ \kwto\ \vleft\semi\\
    \ \ \kwrecv\ x\ \kwfrom\ \vright\semi\\
    \ \ \vsum\ \coleq\ \vsum+x\semi\\
    \ \ i\ \coleq\ i+1\semi\\
    \rb\\
    \kwassert\ \vsum=\sum_{i=0}^{\vNP-1}A[i]\semi
  \end{program}
  \caption{\code{cycsum}: cyclic sum reduction to all processes.
    Each process stores a real number in $x$, then repeatedly sends to
    its left, receives from its right, and adds the received value to
    $\vsum$.  At the end, each process holds the sum of the original
    $x$ values.}
  \label{fig:cycsuma}
\end{figure}

An example \mplite{} program is shown in Fig.\  \ref{fig:cycsuma}.  The
goal of this program is to compute the sum, over all processes $i$,
of the value of variable $x$ on process $i$.  The result is stored in
variable $\vsum$ on every process.   The goal is achieved by cycling
the $x$ values around a ring as each process accumulates their sum.

Lines 1, 3, and 14 specify the expected behavior of the program.  A
parameter $A$ is assumed to be an array of reals of length $\vNP$.
This is the same array on every process, as $A$ is a parameter.
Furthermore the variable $x$ is assumed to be element $\vpid$ of $A$.
The assertion claims that \textit{sum} is the sum of the elements of
$A$.  The program is correct.  Note correctness requires that $+$ is
associative and commutative, so would not necessarily hold for
floating-point arithmetic.

\subsection{The transformation}

We now show how to transform an \mplite{} program to a
sequential program in such a way that the correctness of the
sequential program implies the correctness of the original program.
The sequential language is the same as the parallel one, except there
are no \kwsend{} or \kwrecv{} statements, and \vpid{} and \vNP{} are
not reserved.  The result of this transformation on \code{cycsum}
is shown in Fig.\  \ref{fig:transformb}.

\newcommand{\vsc}{\ident{sc}}
\newcommand{\vrc}{\ident{rc}}
\newcommand{\vlevel}{\ident{level}}

\begin{figure}
  \(
  \begin{array}[t]{l@{\hspace{0.5ex}}c@{\hspace{1.5ex}}l}
    \nummsg(s,d)
    &\eq& d=(s+\vNP-1)\%\vNP\; ?\; \vNP-1 : 0 \semi\\
    \msginv(s,d,j,t) &\eq& t=A[(s+j)\%\vNP] \semi
  \end{array}
  \)
  \begin{flushleft}
    (a) Number of messages from $s$ to $d$ and message invariant.
  \end{flushleft}
  \vspace{2.5ex}

  \begin{center}
    \includegraphics[scale=.75]{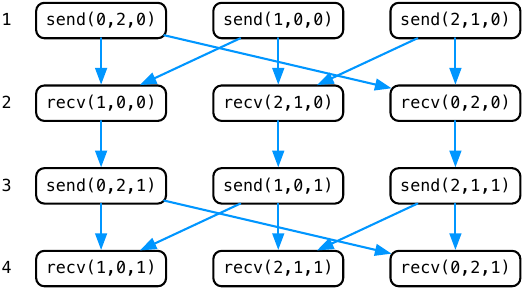}
  \end{center}
  \begin{flushleft}
    (b) Happens-before relation with levels, $\vNP=3$.  $\HB$ is
    the transitive closure of the relation shown.
  \end{flushleft}
  \vspace{-.5ex}
    
  \[
    \begin{array}{rcl}
      \slevel(s,d,j) & \eq & 2*j+1 \semi\\
      \rlevel(s,d,j) & \eq & 2*j+2 \semi
    \end{array}
  \]
  \begin{flushleft}
    (c) Level of each send, level of each receive.
  \end{flushleft}
  \vspace{-1.5ex}
  
  \[
    \fun{AI}(x,y)\ \equiv\ \kwassert\ x<y\semi\ x\ \coleq\ y\semi
  \]
  \begin{flushleft}
    (d) Macro to check assignment increases value.
  \end{flushleft}
  \caption{Information used to transform \code{cycsum}.}
  \label{fig:transforma}
\end{figure}

\begin{figure}
  \hspace{2ex}
  \begin{program}
    \kwparam\ \vNP\comma\ A\semi\\
    \begin{array}[t]{@{}l@{}}
      \kwvar\ \vpid\comma\ i\comma\ x\comma\ \vsum\comma\ \vleft\comma\
      \vright\comma\ \vsc\comma\ \vrc\comma\ \vlevel\semi
    \end{array}\\
    \begin{array}[t]{@{}l@{}}
      \kwassume\ 0\leq\vpid<\vNP\ \wedge\ \\
      \ \ \visRealArray(A,\vNP)\ \wedge\ 
      x=A[\vpid]\semi
    \end{array}\\
    \kwassert\ \fun{SpecOK}\semi\ \ \emph{// see \eqref{eq:specok}}\\
    \vsc\ \coleq\ \fun{zeroIntArray}(\vNP)\semi\\
    \vrc\ \coleq\ \fun{zeroIntArray}(\vNP)\semi\\
    \vlevel\ \coleq\ 0\semi\\
    \vleft\ \coleq\ (\vpid+\vNP-1)\%\vNP\semi\\
    \vright\ \coleq\ (\vpid+1)\%\vNP\semi\\
    i\ \coleq\ 1\semi\\
    \vsum\ \coleq\ x\semi\\
    \kwwhile\ \lp\ i<\vNP\ \rp\ \lb\\
    \ \ \fun{AI}(\vlevel, \slevel(\vpid, \vleft, \vsc[\vleft]))\semi\\
    \ \ \kwassert\ \msginv(\vpid, \vleft, \vsc[\vleft], x)\semi\\
    \ \ \vsc[\vleft]\ \coleq\ \vsc[\vleft] + 1\semi\\
    \ \ \fun{AI}(\vlevel, \rlevel(\vright, \vpid, \vrc[\vright]))\semi\\
    \begin{array}[t]{@{}l@{}}
      \ \ \kwchoose\ x\ \st\ \msginv(\vright, \vpid, \vrc[\vright], x)\semi
    \end{array}\\
    \ \ \vrc[\vright]\ \coleq\ \vrc[\vright]+1\semi\\
    \ \ \vsum\ \coleq\ \vsum+x\semi\\
    \ \ i\ \coleq\ i+1\semi\\
    \rb\\
    \kwassert\ \vsum=\sum_{i=0}^{\vNP-1}A[i]\semi\\
    \begin{array}[t]{@{}l@{}}
      \kwassert\ \forall j\ .\ 0\leq j<\vNP\ra\\
      \ \ \vsc[j]=\nummsg(\vpid,j)\ \wedge\ \vrc[j]=\nummsg(j,\vpid)\semi
    \end{array}
  \end{program}
  \caption{Result of transforming \code{cycsum} to sequential.}
  \label{fig:transformb}
\end{figure}

\subsubsection{The message invariant}
Like the parallel program, the sequential program represents the
execution of a single arbitrary process.  It declares a parameter
$\vNP$ and a variable $\vpid$ and adds the assumption
\( 0\leq\vpid<\vNP\).  To transform the send and receive instructions,
we ask the user to provide information on the messages that are sent
in an execution of the program.  First, the user provides a function
$\nummsg(s,d)$, which specifies the number of messages sent from
process $s$ to process $d$ ($s,d\in\PID$).  For \code{cycsum}, this
function is shown in Fig.\  \ref{fig:transforma}(a).  If $d$ is the
left neighbor of $s$, $\vNP-1$ messages are sent; otherwise $0$
messages are sent.  This function must evaluate the same on every
process, and at any point during the execution.  To ensure this, the
definition is required to be \emph{universal}, i.e., it cannot use
variables, only literal constants, program parameters, and the formal
parameters ($s$ and $d$).
%but only (program or formal) parameters and literal constants.

Next, the user defines a universal predicate $\msginv(s,d,j,t)$ which
holds when $t$ is the $j^{\textit{th}}$ message sent from $s$ to $d$.
For the example, \msginv{} specifies the exact value of $t$, given
$s$, $d$, and $j$.  This choice allows us to prove the full functional
correctness of the program, i.e., the assertion on line 14.  Other
choices are possible: e.g., if one wanted to prove only
deadlock-freedom, one could use the predicate $\T$.  Or, one might
assume that all $x$'s lie within some interval $[a,b]$ and prove the
resulting sum is in $[a*\vNP,b*\vNP]$, in which case the message
invariant could be $a\leq t\leq b$.

To utilize the message invariant, we add two variables to the
sequential program: $\vsc$ (``send count'') and $\vrc$ (``receive
count'').  Each of these stores an array of integers of length $\vNP$,
initially all $0$s.  The former keeps track of the number of messages
sent from this process to other processes; the latter the number of
messages received.

Now a statement of the form $\kwsend\ x\ \kwto\ d\semi$ is replaced
with
\[
  \kwassert\ \msginv(\vpid, d, \vsc[d], x)\semi\ 
  \vsc[d]\ \coleq\ \vsc[d]+1\semi.
\]
This checks that the messages sent by the program satisfy the given
message invariant.  A statement of the form
$\kwrecv\ x\ \kwfrom\ s\semi$ is replaced with
\[
  \kwchoose\ x\ \st\ 
  \msginv(s, \vpid, \vrc[s], x)\semi\ 
  \vrc[s]\ \coleq\ \vrc[s]+1\semi.
\]
Here we use the fact that the $i^{\textit{th}}$ message sent from $s$
to $\vpid$ is the $i^{\textit{th}}$ message received in $\vpid$ from
$s$ --- the FIFO semantics of channels. These two substitutions
function as a kind of assume-guarantee reasoning
\cite{dingel:2003:assume-guarantee}.

Finally, the assertion at line 23 in Fig.\  \ref{fig:transformb} is
added to the sequential program.  This checks that the numbers of
messages sent and received from this process to other processes agrees
with the user's \nummsg{} specification.

\subsubsection{Levels}
The transformation thus described is incomplete.  To see why, consider
the parallel program obtained by transposing lines 9 and 10 in Fig.\ 
\ref{fig:cycsuma}.  Then every process attempts to receive before
sending, which clearly results in deadlock.  Nevertheless, the
transformed sequential program so far described is correct.  The
problem is the sequential program fails to check that communication
events from different processes can be interleaved into a coherent
global order.

The abstraction that is useful for reasoning about this is the
\emph{happens-before} relation \HB{} of an execution.  An execution
generates a set of communication events, each of which is either a
send or receive, and is identified by a source, destination, and
index.  \HB{} is a binary relation on this set, the transitive closure
of the union of two relations: (1) the \emph{intra-process order},
consisting of all $u\ra v$ where $u$ and $v$ are communication events
from the same process and $u$ occurs before $v$, and (2) the
\emph{matching relation}, consisting of all $u\ra v$ where $u$ is a
send and $v$ is the receive that receives the message sent by $u$.  It
is not hard to see that the \HB{} of an execution of an \mplite{}
program must be acyclic, as both constituent relations point forward
in time.  The \HB{} relation for a 3-process execution of
\code{cycsum} is shown in Fig.\  \ref{fig:transforma}(b).

Note that in \mplite, the matching relation is completely determined
by the intra-process order: for any $s,d$, the $i^{\textit{th}}$
occurrence of $\kwsend(s,d)$ in process $s$ matches the
$i^{\textit{th}}$ occurrence of $\kwrecv(s,d)$ in process $d$.  Hence
$\HB$ is determined by the execution of each process, independently of
how these events may be interleaved.  In other words, the generic
sequential program determines the $\HB$ relation.  For example, in the
case of the erroneous version of \code{cycsum} where the receives
happen first, the implied $\HB$ of a $2$-process execution would be
the transitive closure of
\begin{center}
  \includegraphics[scale=.7]{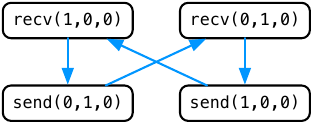}
\end{center}
But this relation has a cycle, and therefore cannot correspond to
an execution of the parallel program.
      
We therefore require the user to provide evidence that $\HB$ is
acyclic.  It is well-known that a directed graph is acyclic if and
only if each node $u$ can be assigned a positive integer \emph{level},
$l(u)$, in such a way that for any edge $u\ra v$, $l(u)<l(v)$.  We ask
the user to provide two universal functions that together specify the
level of every communication event.  Function $\slevel(s,d,j)$
specifies the level of the $j^{\textit{th}}$ send operation from $s$
to $d$, while $\rlevel(s,d,j)$ specifies the level of the
$j^{\textit{th}}$ receive operation in $d$ from $s$.  Generalizing
from the picture in Fig.\  \ref{fig:transforma}(b), appropriate level
definitions for \code{cycsum} are given in Fig.\ 
\ref{fig:transforma}(c).

The level functions are used in the transformation as follows.  First,
an assertion checks that both functions are positive and that each
send has a lower level than its matching receive:
\begin{multline}
  \fun{SpecOK} \equiv\ \forall s,d,j\ .
  0\leq s,d<\vNP\ \wedge\ 
  0\leq j<\nummsg(s,d)\ \ra \\
  0<\slevel(s,d,j)<\rlevel(s,d,j).
  \label{eq:specok}
\end{multline}
Next, a variable \vlevel, initially $0$, is added to the program, to
keep track of the current level, which is the level of the last
communication event from this process.  At a communication event, the
level is updated using the provided specifications---but only after
checking that the new level is strictly greater than the current one.
For brevity, we use a macro $\fun{AI}(x,y)$, which asserts $x<y$ and
assigns $y$ to $x$.  For a statement of the form
$\kwsend\ x\ \kwto\ d\semi$, we insert the code
\[
  \fun{AI}(\vlevel, \slevel(\vpid, d, \vsc[d]))\semi.
\]
For $\kwrecv\ x\ \kwfrom\ s\semi$, the following is inserted:
\[
  \fun{AI}(\vlevel, \rlevel(s, \vpid, \vrc[s]))\semi.
\]
Together with the check in $\fun{SpecOK}$, these assertions verify
that $\HB$ is acyclic.

This completes the description of the transformation.  In summary, the
user provides a parallel \mplite{} program together with four
universal function definitions, and the algorithm outputs a sequential
program.  The following is proved in \proofloc.

\begin{theorem}
  \label{thm:main}
  Let $P$ be an \mplite{} program and $S$ the sequential
  program resulting from transforming $P$.  If $S$ is correct
  then $P$ is correct.
\end{theorem}

\subsection{Variations}
\label{sec:var}

\subsubsection{Synchronous semantics}
\label{sec:synch}
The material above uses an asynchronous semantics for $\mplite$, but
we could just as well choose the synchronous semantics.  In this
interpretation, there are no buffers.  A send blocks until control in
the receiving process is at a matching receive, at which point the two
operations may execute together in a single atomic step.  Everything
goes through as before, with the following modifications:
\begin{itemize}
\item Each communication event comprises a matching send and receive
  pair, so involves two processes.  The $\HB$ relation of an execution
  is the transitive closure of the relation consisting of all
  $u\ra v$, where $u$ and $v$ involve a common process $p$, and the
  $p$-operation in $u$ precedes the $p$-operation in $v$.
\item The user only needs to provide a single level function, say
  $\slevel(s,d,j)$, specifying the level of the event comprising the
  $j^{\textit{th}}$ send from $s$ to $d$ and the $j^{\textit{th}}$
  receive in $d$ from $s$.
\item In equation \eqref{eq:specok}, remove the ``$<\rlevel(s,d,j)$''.
\item In the transformation of a \kwrecv{} statement, use \slevel{} in
  place of \rlevel{}.
\end{itemize}

Note that \code{cycsum} deadlocks under the synchronous semantics,
as all processes first send, then receive.  This situation is so
common that MPI provides a procedure (\mpisendrecv) that lets
a send and a receive operation execute concurrently.   We can
do the same in \mplite{} by adding a statement of the form
\[
  \kwsend\ x\ \kwto\ d\ \bbar\ \kwrecv\ y\ \kwfrom\ s \semi.
\]
The semantics are the same as executing the two operations
sequentially, only now the two can complete in either order.
This statement is transformed to
\begin{center}
  \begin{program}
    \kwassert\ \slevel(\vpid, d, \vsc[d]) > \vlevel\semi\\
    \kwassert\ \slevel(s, \vpid, \vrc[s]) > \vlevel\semi\\
    \vlevel\ \coleq\ \fun{MAX}(\slevel(\vpid, d, \vsc[d]),
    \rlevel(s, \vpid, \vrc[s]))\semi\\
    \kwassert\ \msginv(\vpid, d, \vsc[d], x)\semi\\
    \kwchoose\ y\ \st\ \msginv(s, \vpid, \vrc[s], y)\semi\\
    \vsc[d]\ \coleq\ \vsc[d] + 1\semi\ 
    \vrc[s]\ \coleq\ \vrc[s] + 1\semi.
  \end{program}
\end{center}
Using this combined statement in place of the send and receive
statements in \code{cycsum} yields a correct synchronous program
which transforms to a correct sequential program using
level function $\slevel(s,d,j)=j+1$.

\subsubsection{MPI semantics}
MPI does not specify synchronous or asynchronous semantics for the
standard-mode blocking point-to-point operations.  Rather, it says
that any send \emph{may} be buffered \emph{or may} be forced to wait
until the matching receive occurs and the communication takes place
synchronously.  This leaves maximum freedom to the implementation.
While this might seem to complicate verification, in fact it adds
nothing over synchronous semantics.  Specifically, an \mplite{}
program is correct under synchronous semantics iff it is correct under
MPI semantics.  See \cite{siegel-avrunin:2005:ppopp} for proofs.

\section{Verifying C/MPI Programs with \vmfc}
\label{sec:impl}

We have developed a prototype tool for deductive verification of C/MPI
programs using Frama-C.  The Verified MPI extension to Frama-C
(\vmfc{}) is based on the ideas of Section \ref{sec:transform}, but
introduces new concepts to deal with the many ``accidental
complexities'' in MPI and C.  It comprises a Perl program that parses
and transforms the C source, and an implementation of \code{mpi.h}.
We illustrate \vmfc{} on \code{cycsum.c}, a C/MPI version of the
\mplite{} program \code{cycsum} of Fig.\  \ref{fig:cycsuma}.  The
primary procedure in this program is named \code{sum}, which consumes
a \code{double} $x$ and returns the sum of the $x$ values over all
processes.  It is implemented exactly as in Fig.\  \ref{fig:cycsuma}
but uses $\mpisendrecv$ (see Sec.\ \ref{sec:var}).

\subsection{Background on MPI and Frama-C}

\subsubsection{MPI}
The MPI Standard \cite{mpi-forum:2015:mpi31} specifies a
message-passing library with bindings in C or Fortran.  We assume
basic familiarity with MPI.  In this paper, we use the C bindings, and
for now limit attention to the standard-mode blocking point-to-point
procedures

\begin{small}
\begin{verbatim}
MPI_Send(buf, count, datatype, dest, tag, comm);
MPI_Recv(buf, count, datatype, source, tag, comm, status);
MPI_Sendrecv(sendbuf, sendcount, sendtype, dest, sendtag,
  recvbuf, recvcount, recvtype, source, recvtag, comm, status);
\end{verbatim}
\end{small}

\noindent{}(Later we will see how to incorporate blocking collective
procedures.)  We use only the predefined communicator \mpicommworld,
comprising all processes, and consider the rank in \mpicommworld{} to
be the PID.  We prohibit \mpianysource, but allow \mpianytag.

Standard procedures for initializing and finalizing
MPI, and for obtaining the rank of a process and the number of
processes, are also in the supported subset.

\subsubsection{Frama-C}

Frama-C \cite{cuoq-etal:2012:frama-c, kosmatov-etal:2024:frama-c} is
an analysis framework for C programs with a plug-in-based design.  The
\Wp{} plug-in \cite{baudin-etal:2025:wp-manual,
  blanchard-etal:2024:wp} is used for deductive verification.  It
accepts a C program with annotations, specially formatted comments of
the form \code{/*@} \ldots \code{*/} or \code{//@ \ldots}.
Annotations are in the ACSL specification language
\cite{baudin-etal:2025:acsl} and are used to write procedure
contracts, including preconditions, postconditions, and frame
conditions.  They can also encode loop invariants and variants,
assertions, and assumptions (\code{admit} \emph{expr}\code{;}).
Annotations can also contain \emph{ghost code}, executable code which
appears only during verification.  Ghost code can also be used to add
ghost parameters to a procedure, and ghost arguments in procedure
calls.  Logical functions and predicates can also be defined
(including inductively) in annotations, and used in other parts of the
specification.

From the annotated code, Frama-C generates verification conditions
(VCs) whose validity implies that procedures satisfy their contracts
and all assertions hold.  The VCs can be discharged with a number of
automated and interactive theorem provers.  In the case study reported
here, the automated provers Alt-Ergo \cite{ocamlpro:2025:alt-ergo}, Z3
\cite{demoura:2008:z3}, and CVC5 \cite{barbosa-etal:2022:cvc5}
discharge all VCs.

\subsection{\vmfc{} Procedure Contracts}

\vmfc{} supports new classes of ACSL annotations, all of which begin with
the keyword \code{mpi}.  It includes a bespoke version of the header
file \code{mpi.h}, which defines a number of primitives available in
annotations.  These are all prefixed by \code{VM}, and include
\code{VM\U{}NP} (the number of processes) and \code{VM\U{}pid} (the
process ID).  \vmfc{} takes a C/MPI program with these annotations,
converts it to an ordinary sequential ACSL-annotated C program, and
invokes Frama-C/\Wp{}.

\begin{figure}
  \begin{small}
\begin{verbatim}
/*@ mpi collective;
    mpi universal A[0 .. VM_NP-1];
    requires Init && VM_state==VM_Active && rank==VM_pid &&
       nprocs==VM_NP;
    requires Array(A, VM_NP) && x==A[VM_pid];
    ensures E1: \result == Cyclic_Sum(A, VM_NP, 0, VM_NP);
    ensures E2: VM_UnchangedGlobal{Pre,Here};
    assigns VM_GlobalVars;   */
T sum(T x) /*@ ghost (\ghost const T * A) */  ...
\end{verbatim}
  \end{small}
  \caption{Contract for procedure \code{sum} in \code{cycsum.c}.}
  \label{fig:contract}
\end{figure}

The contract for \code{sum} is shown in Fig.\  \ref{fig:contract}.  It
begins with the clause \code{mpi collective;}, indicating that
$f$ is intended to be called collectively (by all processes) and that
$f$ is \emph{communication-closed}, i.e., messages sent within $f$
should be received in $f$.  Such a procedure will be transformed
following the approach described in Section \ref{sec:transform}.

A collective procedure contract may also contain a clause of the form
\[
  \code{mpi universal}\ \textit{tset}\ \gtxt{(} \comma\ \textit{tset}
  \gtxt{)*} \semi
\]
A \emph{tset} is an ACSL construct representing a set of memory
locations.  ACSL provides a rich language for tsets, but for our
purposes, each is either a variable or an expression of the form
\code{a[}\textit{lo}\code{..}\textit{hi}\code{]}, where \code{a} is a
variable and \textit{lo} and \textit{hi} are integer expressions.  The
latter represents a section of an array.  These variables may be
formal parameters of $f$ or global variables.  The clause encodes the
requirement that when $f$ is called, the listed memory locations
are universal, i.e., have the same value on every process.  It also
encodes the guarantee that $f$ will not assign to these locations.
Hence, these variables (and array entries) play the role of
``parameters'' in the sense of Section \ref{sec:transform}: they can
be used in a message specification.

In the case of \code{sum}, a ghost parameter \code{A} is added, which
plays the role of array $A$ in Fig.\  \ref{fig:cycsuma}.  The contract
requires that \code{A} is a valid array of doubles of length \vmnp.
The elements of this array are required to be universal.  The contract
also requires that MPI has been initialized (by a call to \mpiinit),
but not yet finalized (\mpifinalize).  Finally, the contract uses a
logic function \verb!Cyclic_Sum!, defined inductively to satisfy
\begin{equation}
  \code{Cyclic\U{}Sum}(p,n,s,m) = \sum_{i=0}^{m-1}p[(s+i)\%n].
\end{equation}

\subsection{MPI Regions}

Within the body of $f$, a section of code that implements a
self-contained collective routine is called an \emph{MPI region}.  A
region corresponds to a ``program'' in the sense of Section
\ref{sec:transform}.  The boundaries of a region are bracketed with a
\emph{begin region} annotation and the \emph{end region} annotation
(\code{//@ mpi end region;}).  The \emph{begin region} annotation used
in procedure \code{sum} in \code{cycsum.c} is shown in Fig.\ 
\ref{fig:beginregion}.

\begin{figure}
  \begin{small}
\begin{verbatim}
/*@ mpi begin region 1:
    nummsg(src,dest) =
      (dest==(src+VM_NP-1)%VM_NP ? VM_NP-1 : 0);
    mcount(src,dest,idx) = 1;
    mdtype(src,dest,idx) = MPI_DOUBLE;
    msgtag(src,dest,idx) = TAG;
    msginv(src,dest,idx,buf,count,dt) =
      *(buf)==A[(src+idx)%VM_NP];
    slevel(src,dest,idx) = idx+1;  */
\end{verbatim}
  \end{small}
  \caption{The \emph{begin region} annotation for the MPI region of
    \code{sum} in \code{cycsum.c}.}
  \label{fig:beginregion}
\end{figure}

The user assigns an ID number ($1$) to the region.  A procedure $f$
can have multiple regions, each with an ID unique within $f$.  This is
followed by several universal function definitions, including
$\anummsg$, $\amsginv$, and $\aslevel$.  These follow the descriptions
in Section \ref{sec:transform}, except that $\amsginv$ consumes a
pointer, count, and datatype in place of a single argument, as these
three arguments together are needed to specify the message data.
Functions \code{mcount}, \code{mdtype}, and \code{msgtag} specify the
expected count, datatype, and tag of each message; as with the message
data, a send will assert this metadata is as expected, and each
receive will rely on these being true of an incoming message.

A transformed receive asserts that the count $c$ of the incoming
message is at most the receive's count argument; that the datatypes
match exactly; and that the tags are either equal or the tag argument
is \mpianytag.  It then ``havocs'' the first $c$ elements of the
receive buffer and assumes (using \code{admit}) that the message
invariant holds over the havoced region.  The current level, stored in
variable \vmlevel, is updated as explained in Section
\ref{sec:transform}.

\vmfc{} requires that all MPI communication takes place within some
region.  There are two kinds of regions: the kind described above,
explicitly demarcated with annotations, is an ``internal'' region.
But $f$ may also call other (user-defined or MPI-defined) collective
procedures.  \vmfc{} considers each call to a collective procedure to
be an ``external'' region unto itself.  In fact, there is little
difference between the two kinds of regions: both are
communication-closed routines invoked by all processes, specified by
either a message specification or a procedure contract.  Regions
cannot be nested: so, for example, it would be erroneous for a
collective call to occur within an internal region.

\vmfc{} enforces a key requirement regarding regions in a procedure $f$: on
any execution of $f$, all processes must execute the same region
sequence.  For example, if one process first executes internal region
1, then calls a collective procedure $g$, and finally executes
internal region 2, then all processes must do so, in the same order.
This is a natural generalization of the requirement in the MPI
Standard that all processes invoke collective procedures in the same
order.  Without this restriction, it would be nearly impossible to
avoid deadlocks, or messages ``leaking'' from one region into another.

The constraint that all processes execute the same region sequence can
be viewed as a ``universal'' constraint on control-flow, analogous to
the universal constraint on data required by many procedures.  To
enforce it, we ask the user to define universal functions specifying
the sequence.  An annotation of the form

\begin{small}
  \[
    \begin{array}{@{}l@{}}
      \code{/*@ mpi begin regions: nregions = }\textit{univ-expr}
      \code{;}\\
      \code{\ \ \ \ region(i) = }\textit{univ-expr}\code{;}\ 
      \ldots\ \code{*/}
    \end{array}
  \]
\end{small}

\noindent{}must occur once in $f$, before any MPI regions, and

\begin{small}
\begin{verbatim}
     //@ mpi end regions;
\end{verbatim}
\end{small}
occurs after all regions.  Function $\nregions$ specifies the number
of regions, $\region(i)$ specifies the ID of the $i$-th region.
Within the definition of $\region(i)$, the name of a collective
procedure $g$ is used as the ID of an external region, while any
positive integer expression is interpreted as the ID of an internal
region.  Procedure \code{sum} of \code{cycsum.c} is particularly
simple: control passes through exactly one region, internal region
$1$, so the begin regions annotation is simply

\begin{small}
\begin{verbatim}
/*@ mpi begin regions: nregions = 1; region(i) = 1; */
\end{verbatim}
\end{small}

We have seen that a collective procedure $g$ may require that certain
variables are universal when $g$ is called.  This must be enforced at
the call site.  Again, we require the user to provide the evidence of
universality.  For every collective $g$ called by $f$, and every
variable \code{x} declared \code{universal} in the contract of $g$,
the user must provide a definition of the form
\(
  \code{g\#x(i) = } \textit{univ-expr} \code{;}
\).
This specifies the expected value of \code{x} when $\region(i)=g$.
For an array variable $A$, the definition has the form
\[
  \code{g\#A(i,j) = } \textit{univ-expr} \code{;}
\]
defining element $j$ of $A$ when $\region(i)=g$.
These definitions go in the \emph{begin regions} annotation of $f$.
The transformation adds state that keeps track of the current region,
checks that all processes complete the expected region sequence, and
asserts the arguments or globals expected to be universal at a call
match the specified values.

Program \code{cycsum.c} contains a procedure \code{test1} that
performs a simple test of \code{sum}.  The contract for \code{test1}
assumes $A$ is a universal array with $A[j]=j$ for $0\leq j<\vmnp$.
This procedure contains a single region, which is a call to
\code{sum}.  The \emph{begin regions} annotation is:

\begin{small}
\begin{verbatim}
/*@ mpi begin regions: nregions = 1;
    region(i) = sum; sum#A(i,j) = j; */
\end{verbatim}
\end{small}

% We cannot show the complete \code{cycsum.c} for lack of space, but
% here summarize a few of the remaining salient aspects.
The annotations also include a mini-theory of cyclic sums: in addition
to the inductive definition, there are three lemmas, one ``lemma
macro'' proving that a cyclic sum does not change if its arguments
don't, and a ``lemma function'' \cite[\S4.4.7]{blanchard-etal:2024:wp}
that proves, for $0\leq s<n$,
\begin{equation}
  \code{Cyclic\U{}Sum}(a, n, s, n) = \code{Cyclic\U{}Sum}(a, n, 0, n).
  \label{eq:cyclicLemma}
\end{equation}
The program is proved using option \code{-wp-model real} to interpret
floating-point operations as real arithmetic.  Neither
\eqref{eq:cyclicLemma} nor the contract hold for floating-point
arithmetic, and in our view a real arithmetic specification is
appropriate for code that performs a reduction.

The code contains two intermediate assertions in addition to a call to
the lemma function and to the lemma macro.  The \code{for} loop has
$7$ loop invariants; some of these refer to variables inserted by the
transformation, such as \code{VM\U{}rc} (the receive count array) and
\code{VM\U{}level}.

\subsection{Evaluation}

We have used \vmfc{} to verify $5$ C/MPI programs.\footnote{All
  artifacts for this study are available at
  \cite{siegel:2026:param-ver-artifact}.}  Verification includes
functional correctness and freedom from deadlock, MPI errors such as
buffer overflows or inconsistent collective call sequences, and
standard C run-time errors such as out-of-bound indexes or integer
overflows.  No bounds are placed on the program parameters other than
those implied by C's integer types.  All programs also compile without
warnings or errors, and run a simple test when executed with any
number of processes.  The programs are:
\begin{enumerate}
\item \code{cycsum.c}: cyclic sum reduction described in Section
  \ref{sec:impl}.
\item \code{allsum.c}: like \code{cycsum}, but uses
  a star network instead of a ring.
\item \code{diffuse1d.c}: $1$-dimensional diffusion (or ``heat
  equation'') simulation with fixed boundary condition.  Ghost cell
  exchanges take place in two directions.  Parameterized by the number
  of discrete spatial points (\code{NX}) and the number of time steps
  (\code{NSTEP}), in addition to the number of processes.
\item \code{bcast.c}: a simple implementation of MPI's broadcast
  collective \mpibcast.
\item \code{gather.c}: implementation of \mpigather.
\end{enumerate}
We used the real model of arithmetic for 1--3, and floating-point for
4--5.

\begin{figure}
  \centering
  \setlength{\tabcolsep}{6pt}
  \begin{tabular}{lrrrrr}
    \multicolumn{1}{c}{Program}
    & \multicolumn{1}{c}{ACSL}
    & \multicolumn{1}{c}{Code}
    & \multicolumn{1}{c}{VCs}
    & \multicolumn{1}{c}{TO}
    & \multicolumn{1}{c}{Time} \\
    \hline
    \code{cycsum.c}    & 132 & 38 & 166 & 60 & 60 \\ % timeout: 60.  VCs:65+67=132
    \code{allsum.c}    &  79 & 51 & 192 & 15 & 18 \\ % timeout: 15.
    \code{diffuse1d.c} & 125 & 55 & 238 & 15 & 56 \\ % timeout: 15.
    \code{bcast.c}     &  78 & 46 & 144 & 15 &  9 \\ % timeout: 15.
    \code{gather.c}    & 140 & 53 & 220 & 30 & 43 \\ % timeout: 30
  \end{tabular}
  \caption{Statistics for proved C/MPI programs: number of lines
    containing ACSL annotations, number of lines containing C code,
    number of verification conditions generated, timeout (seconds)
    specified for each prover call, total verification time.}
  \label{fig:stats}
\end{figure}

Fig.\ \ref{fig:stats} shows the number of lines containing
(VM-extended) ACSL annotation and the number of lines containing C
code.  The annotation burden is high, $1.5$--$3.5\times$ the amount of
C code, but is not unusual for this style of C code verification.  For
\code{cycsum.c}, about half of the ACSL is for the mini-theory of
cyclic sums, in a separate file which could conceivably be reused in
other settings.

The remaining columns are the number of VCs generated, the timeout for
each prover call, and the total verification time, run on a 2025 M4
MacBook Air.  All verification times are under one minute.  Prover
calls dominate these times.  Transformation, compile, and test
execution times are all negligible.

In \code{diffuse1d.c}, an inductively defined logic function gives the
temperature at any time and point.  This is used throughout the
specification.  A procedure \code{exchange} performs the ghost-cell
exchange with two \mpisendrecv{} calls.  A procedure \code{run} loops
over time and at each iteration calls \code{exchange} and then a
procedure that performs the local update.  Its \emph{begin regions}
annotation contains:

\begin{small}
\begin{verbatim}
nregions = NSTEP; region(i) = exchange; exchange#time(i) = i;
\end{verbatim}
\end{small}

Program \code{bcast.c} defines a procedure \code{Bcast} with (almost)
the same signature as MPI's \mpibcast.  Like many of the MPI
collectives, the MPI Standard requires that several of the arguments
be the same on all processes.  This is indicated in our contract,
which reads in part:

\begin{small}
\begin{verbatim}
mpi universal count, datatype, root, Data[0..count-1];
requires VM_pid == root ==> Equal(buf, Data, count);
ensures Equal(buf, Data, count);
\end{verbatim}
\end{small}

Here, \code{Data} is a universal ghost parameter, and \code{Equal} is
a logic predicate stating that two arrays are pointwise equal.  The
contract says that if the buffer on the root process contains
\code{Data} in the pre-state, then in the post-state the buffer on any
process contains \code{Data}.  This captures the semantics in a
concise natural way.

The one change in signature is that in MPI, the \code{buf} argument
has type \code{void*}.  \vmfc{} requires all buffers to have a type of
the form pointer-to-$T$ for a basic type $T$ (e.g., \code{double},
\code{int}).  This comes down to limitations in the \Wp{} plug-in for
reasoning about pointers.  The program therefore includes concrete
definitions for $T$ and the corresponding MPI datatype.  These can be
easily changed to verify using different element types, but at this time
we do not know a way to verify all types at once.  This limitation
applies to other MPI collective procedures.

\section{Related Work}
\label{sec:related}

The subset of MPI used in this paper was analyzed in
\cite{siegel-avrunin:2005:ppopp}, where a number of basic facts
concerning determinism are proved.  For example, for a given input,
freedom from assertion violations can be verified by examining a
single execution (Theorem 6), and deadlock-freedom can be verified by
showing a single synchronous execution is deadlock-free (Corollary 7).

Finite-state verification techniques, combining model checking and
symbolic execution, have been used in several tools to verify MPI
programs within a process bound.  Notable examples are MPI-Spin
\cite{siegel:2007:vmcai}, CIVL \cite{siegel-etal:2015:civl_sc}, Hermes
\cite{khanna-etal:2018:dynamic-symbolic-mpi}, and MPI-SV
\cite{yu-etal:2020:mpi-sv}.  Hermes has scaled to 32 processes on some
deterministic programs, and up to $8$ on ones using \mpianysource.
The technique in this paper proves correctness without any process
bound, but does not handle \mpianysource.

ParTypes \cite{vasconcelos-etal:2022:partypes, lopez-etal:2015:mpi}
was one of the first tools for parameterized verification of
deterministic MPI programs.  It requires the user to write a protocol
describing communication operations in an execution from a global
point of view.  Type-checking techniques are used to prove each
process conforms to the protocol. The protocols do not deal with the
content of messages, so the approach is not used to verify functional
correctness, but can prove deadlock-freedom.  The communication
pattern used in the program must be expressible in the protocol
language, which has built-in common patterns such as \emph{gather},
\emph{scatter}, \emph{reduce}, and \emph{broadcast}.  The approach of
this paper does not require a protocol, but instead simpler
abstractions: a global specification of the ``level'' of each
send/receive operation in the happens-before relation of a collective
region, and the global sequence of regions.  In addition, our approach
enables reasoning about the data \emph{via} a message invariant.

A procedure contract system extending ACSL for MPI was explored in
\cite{luo-siegel:2024:contracts}.  This system includes a term of the
form $\code{\bs{}on(}e\code{,}i\code{)}$, which returns the result of
evaluating expression $e$ on process $i$.  While very general, there
is no obvious way to translate such an expression into a form that
could be used by Frama-C.  Instead, \cite{luo-siegel:2024:contracts}
uses the CIVL model checker to verify procedures conform to their
contracts for a bounded process count.  In contrast, we use a more
restricted construct, \code{universal}, which declares a variable has
the same value on every process, and requires that only universal
variables can be used in certain expressions.  Everything in our
extension can be easily translated to an ordinary ACSL-annotated
sequential C program, which can be verified by Frama-C in full
generality.

Using message invariants in assume-guarantee reasoning
for channel communication has been explored before; cf.\
\cite{leino-mueller-smans:2010:deadlock}.

As discussed in Sec.\ \ref{sec:intro}, this paper builds on work on
DafnyMPI \cite{fedchin-etal:2026:dafnympi}, but improves in several
ways.  First, DafnyMPI requires the parallel algorithm to be
implemented in the language of a verification tool while \vmfc{} works
off of existing C/MPI code.

Second, DafnyMPI requires the tags to reflect a global message order,
essentially requiring the user to provide a topological sort on the
\emph{happens-before} relation.  The MPI version of \code{cycsum}, for
example, might use the tag $i*\vNP+\vpid$.  \vmfc{} has no
restrictions on the tags and takes a ``local'' view: for each $s,d$,
messages are identified by the order they are sent from process $s$ to
process $d$. \vmfc{} requires only the minimal information needed to
show \emph{happens-before} is acyclic: the \emph{level} of the $j$-th
message from $s$ to $d$.  For \code{cycsum}, this expression is $j+1$
(Sec.\ \ref{sec:var}).  This not only reduces the cognitive load on
the user, but reduces the complexity of the expressions that will be
sent to the theorem prover, in this case from quadratic to linear.

While both approaches are procedure-modular, \vmfc{} permits even more
fine-grained modularity.  A procedure body may be decomposed into
separate communication-closed regions, each with its own
specification.  For \code{cycsum}, instead of specifying one region
around the entire procedure, one could specify a region around the
body of the \code{while} loop.  Within each region, a process performs
exactly one send-receive, and the level expression becomes
trivial: the constant function $1$.  The user need only specify the
order that regions are encountered, which must be the same for all
processes.

The region approach also incorporates collective calls in a natural
way: each region may be either internal (a block of code using
point-to-point operations) or external (a call to a collective
function).  In contrast, DafnyMPI numbers the collective calls
separately and requires the user to specify where they are inserted in
the point-to-point message sequence.  Again, the \vmfc{} approach
results in simpler expressions.

On the other hand, the deterministic subset of MPI used in
\cite{fedchin-etal:2026:dafnympi} is more general than that used here:
it consists of the standard-mode \emph{nonblocking} point-to-point
operations \mpiisend, \mpiirecv, and \mpiwait, and some other
operations that can be derived from these, including \mpisend,
\mpirecv, and various collectives.  We are currently working on
extending \vmfc{} to handle nonblocking operations.

Finally, a significant source of complexity in the semantic arguments
of \cite{fedchin-etal:2026:dafnympi} is the ``nondeterminism'' arising
from standard-mode send operations.  For each send, the MPI
implementation may either buffer the message, allowing the send to
complete even before a matching receive has been posted, or force the
sender to wait until the matching receive is called.  This complexity
is unnecessary.  As we pointed out above, to verify freedom from
deadlock and assertion violations, it suffices to consider only the
synchronous semantics; this holds for the nonblocking subset as well
\cite{siegel-avrunin:2007:nbsynch}.  This observation simplifies our
soundness proof considerably, and would also simplify the arguments in
\cite{fedchin-etal:2026:dafnympi}.

\section{Conclusion}
\label{sec:conclude}

We have described an approach for full parameterized verification of
C/MPI programs that use a deterministic subset of MPI.  The user
provides certain functions which specify the sequence of collective
``regions'' through which all processes will pass.  For internal
regions, implemented using point-to-point operations, the user provides
functions specifying the number of messages from each sender to each
receiver and the count, tag, and datatype of each message, and a
constraint on the message contents.  These functions must be
\emph{universal}, expressed using symbols which are known to be
constant and the same on all processes.

We have realized this approach in a prototype tool based on
Frama-C/\Wp{}. The tool takes a C/MPI program annotated in an MPI
extension of ACSL, and produces an ordinary ACSL-annotated sequential
C program that can be verified with Frama-C.  The prototype has some
limitations: buffers must be properly typed (not \code{void*}); MPI
derived datatypes are not yet supported; and for now, only
1-dimensional arrays are supported.  The transformation tool is very
simple and written in Perl; a more robust version will be a proper
Frama-C plug-in in OCaml.

The general approach opens up many avenues for future research.  The
transformed program is a sequential program, and is therefore amenable
to a variety of verification tools and approaches, not just Frama-C
or deductive verification.  The verifier must support choice,
assertions, and assumptions, but these are now standard features.
Model checkers, symbolic execution tools, and static analysis tools
are all reasonable candidates.  Even testing could be improved---it is
easier to test a sequential program than an MPI program.

The big open question is how to achieve parameterized verification of
nondeterministic MPI programs.  This seems to be a very hard problem
that will require completely new ideas.

\subsubsection*{Acknowledgments}

This material is based upon work supported by the U.S.\ National
Science Foundation under Award Number CCF-2446130, and by the U.S.\
Department of Energy, Office of Science, Advanced Scientific Computing
Research (ASCR) Program, under Award Number DE-SC0025953.

This report was prepared as an account of work sponsored by an agency
of the United States Government. Neither the United States Government
nor any agency thereof, nor any of their employees, makes any
warranty, express or implied, or assumes any legal liability or 
responsibility for the accuracy, completeness, or usefulness of any
information, apparatus, product, or process disclosed, or represents
that its use would not infringe privately owned rights. Reference
herein to any specific commercial product, process, or service by
trade name, trademark, manufacturer, or otherwise does not necessarily
constitute or imply its endorsement, recommendation, or favoring by
the United States Government or any agency thereof. The views and
opinions of authors expressed herein do not necessarily state or
reflect those of the United States Government or any agency thereof.

\bibliographystyle{ACM-Reference-Format}
\bibliography{mpinv}

\appendix

\newcounter{artifactCounter}
\newcounter{artpartCounter}
\newcounter{artevalCounter}

\newcommand{\newartifact}{%
    \stepcounter{artifactCounter}%
    \subsection{Computational Artifact $A_{\theartifactCounter}$}%
}

\newcommand{\artrel}{%
    \subsection*{Relation To Contributions}%
}

\newcommand{\artexp}{%
    \subsection*{Expected Results}%
}

\newcommand{\arttime}{%
    \subsection*{Expected Reproduction Time (in Minutes)}%
}

\newcommand{\artin}{%
    \subsection*{Artifact Setup (incl. Inputs)}
}

\newcommand{\artinpart}[1]{%
\subsubsection*{#1}}

\newcommand{\artcomp}{%
    \subsection*{Artifact Execution}
}

\newcommand{\artout}{%
    \subsection*{Artifact Analysis (incl. Outputs)}
}

\newcommand{\apppart}[1]{%
\vspace{2ex}
\begin{center}\begin{large}\textbf{#1}\end{large}\end{center}\par
%\vspace{2mm} % Adjust vertical space as needed
}

\newcommand{\appendixAE}{%
\setcounter{subsection}{0}%
\apppart{Artifact Evaluation (AE)}%
}

\newcommand{\appendixAD}{%
\apppart{Artifact Description (AD)}%
}

\newcommand{\arteval}[1]{%
    \setcounter{artifactCounter}{#1}
    \subsection{Computational Artifact $A_{#1}$}%
}

\newcommand{\artexpl}[1]
{\tagged{explanation}{{\color{colexpl}#1}}}

\newcommand{\artsampl}[1]
{\tagged{example}{{\color{colsample}#1}}}

\section{Artifact Description}

\subsection{Overview of Contributions and Artifacts}

\subsubsection{Paper's Main Contributions}

\begin{description}
\item[$C_1$] A general method to formally verify message-passing
  programs for any number of processes, by transforming to a
  parameterized sequential program using user-provided information
  such as message invariants.  The process is described in the paper,
  where a formal proof of its soundness is also given, in an
  ideadlized model.
\item[$C_2$] An implementation of this method for C/MPI programs using
  the Frama-C analysis and verification platform, including a
  transformer, VMFC, written by us.
\item[$C_3$] A small case study applying VMFC to 5 C/MPI
  programs, showing that they can be proved correct for any number of
  processes and input sizes.  The annotation burden (i.e., the amount
  of annotation comments that must be inserted into the code to enable
  the proof) is large but not atypical for this type of verification
  of C programs.  The verification time is also reasonable.
\end{description}

\subsubsection{Computational Artifacts}

The artifact is on Zenodo \cite{siegel:2026:param-ver-artifact}:

\begin{description}
\item[$A_1$] \url{https://doi.org/10.5281/zenodo.21433793}
\end{description}

\begin{center}
\begin{tabular}{rll}
\toprule
  Artifact ID  &  Contributions &  Related \\
               &  Supported     &  Paper Elements \\
  \midrule
  $A_1$   &  $C_2$ & Section \ref{sec:impl} \\
          &  $C_3$ & Figure \ref{fig:stats}  \\
\bottomrule
\end{tabular}
\end{center}

\subsection{Artifact Identification}

\newartifact

\artrel

This artifact provides the implementation of the prototype
verification tool VMFC (Contribution $C_2$) together with the results
of applying it to 5 C/MPI programs (Contribution $C_3$).

VMFC comprises (1) the Perl script \texttt{vmfc.pl} in directory
\texttt{scripts}, and (2) the header file \texttt{mpi.h} in directory
\texttt{include}.  Together they are used to transform the C/MPI
program with appropriate annotations into a sequential C program with
annotations that enable verification by Frama-C.

Each of the C/MPI programs is in its own directory.  The directories
are:
\begin{enumerate}
\item \texttt{cycsum}
\item \texttt{allsum}
\item \texttt{diffuse1d}
\item \texttt{bcast}
\item \texttt{gather}.
\end{enumerate}

Each directory is structured in the same way.  We describe the case of
\texttt{cycsum}, the others are similar.

Within directory \texttt{cycsum} are the following files:
\begin{enumerate}
\item \texttt{cycsum.c}: this is a complete C/MPI program with
  VMFC annotations added by us.  It is one of the inputs to
  the experiment.
\item \texttt{Cyclic.h}: this is a header file written by us which
  contains a small ACSL theory included by \texttt{cycsum.c}.  It is
  another input to the experiment.  (Not every program has such an
  extra file.)
\item \texttt{cycsum-vm.c}: this is the file resulting from the
  transformation performed by VMFC.  It is a parameterized sequential
  C program with ACSL annotations that can be verified by Frama-C.
  This file is one of the outputs produced by our tool.  When
  reproducing the experiment, the exact same file should be produced.
\item \texttt{out.txt}: this is the output from running Frama-C on
  \texttt{cycsum-vm.c}.  It shows the exact command that invoked
  Frama-C, a summary of the generated verification conditions showing
  they are all discharged, and timing information.
\item \texttt{Makefile}: This can be used to execute the entire
  process: run VMFC on the C/MPI program, then run Frama-C on the
  result.  It can also compile and execute the original C/MPI program
  (which runs a small test).  Simply typing \texttt{make} should do
  all of these things.
\item \texttt{time.txt}: the total time for running \texttt{make}.  
\end{enumerate}

\artexp

The result of each of the 5 experiments should be that all
verification conditions are discharged, so that the program is proved
correct.  Run on an M4 MacBook Air laptop, the entire time should
not exceed 1 minute per program.  The numbers of code lines and
annotation lines are summarized in the paper in Figure 8 and
correspond exactly to the 5 C source files.

\arttime

Artifact Setup: by far the most time will be taken up by installing
and configuring the necessary tools: Frama-C, CVC5, and Z3.  This
might take approximately 1 hour.  Once this has been done, setting up
the artifact is simply downloading it onto the user's local machine.
There should not be any need to adjust anything in the artifact
itself.

Artifact Execution: at most, 5 minutes.

Artifact Analysis: the analysis would consist of reading the output
from Frama-C, checking that the generated \texttt{-vm} files are the
same as those provided in the artifact, and checking the total times.
Approximately 5 minutes.

\artin

\artinpart{Hardware}

Any standard laptop computer with a Unix-style operating system,
including macOS and linux, will do.  No GPU is required.  An MPI
implementation is necessary to run the tests, but that aspect
is not crucial.

\artinpart{Software}

\phantom{x}\vspace{1ex}

\begin{tabular}{lll}
  Frama-C & 31.0 & \url{https://frama-c.com}\\
  Alt-Ergo & 2.6.2 & \url{https://alt-ergo.ocamlpro.com}\\
  CVC5 & 1.3.1 & \url{https://github.com/cvc5}\\
  Z3 & 4.15.4 & \url{https://github.com/z3prover/z3}\\
  Perl & 5.34.1 & \url{https://www.perl.org}
\end{tabular}

Note: for Frama-C and Alt-Ergo, the easiest approach is to install the
OCaml package manager \texttt{opam}, available at
\url{https://opam.ocaml.org}, and use that to install Frama-C and all
of its dependencies.  Instructions are on the Frama-C web page.

Above are the versions we used to run the experiment.  Later versions
will almost certainly work just as well.

\artinpart{Datasets / Inputs}

No datasets are required.  All of the necessary data is inside the
artifact.

\artinpart{Installation and Deployment}

Any C compiler and MPI implementation can be used to compile and
execute these simple programs.  We used Apple clang 17.0.0 and
MPICH 4.3.2.

\artcomp

The \texttt{Makefile} in each directory encodes the workflow.  In each
case: one starts with a C/MPI programs \texttt{foo.c} which has been
annotated with information that will enable its verification.  One
then executes the perl script \texttt{vmfc.pl} on \texttt{foo.c},
which generates \texttt{foo-vm.c}.  Then one invokes Frama-C on
\texttt{foo-vm.c}, which outputs the results to \texttt{stdout}.

In addition, one can compile and execute \texttt{foo.c} in the usual
way, using \texttt{mpicc} and \texttt{mpiexec}.  In the
\texttt{Makefile}, we use 20 processes for this last test.

The main parameter is the ``timeout'', which is an upper bound on the
time that an automated theorem prover is allowed to run on one
verification condition before being killed.  We started with a timeout
of 15 seconds, and if that was not enough to yield a positive result
(due to timeouts), we increased it by 15 seconds, and repeated until a
positive result was achieved.  The maximum timeout used is 60 seconds.

In our tests, there were not big differences in time from repeated
runs.  So we just report the times from a single run.

\artout

For this small study, we manually constructed the table in Figure 8 of
the paper.  For each C file, we simply counted the number of lines
containing any type of ACSL annotation, and the number of lines
containing any type of C code.  (Some lines might contain both.)  The
number of VCs can be read directly from the Frama-C output. The
timeouts were specified by us in the manner described above; different
timeouts might work on different machines.  The total time can be
captured using the command: \texttt{time make}.

\section{Proof of Theorem \ref{thm:main}}
\label{sec:proof}

\begin{proof}
  We start by modifying $P$ by adding variables $\vsc$,
  $\vrc$, and $\vlevel$, and updating these variables immediately after
  each send or receive statement exactly as they are updated in
  $S$.  (We do not add the assertions or the $\kwchoose$
  statement that were inserted into $S$.)  This ``ghost state''
  has no effect on an execution of $P$, nor on its correctness.

  Fix $n\geq 1$ and an assignment of values to the parameters of $P$.
  Let $A$ be the set of states of $P$ with the fixed parameter
  assignment and $n$ processes.  We consider $A$ to be the states of a
  transition system, in which a transition is an atomic step by one
  process, with one exception: a send or receive instruction executes
  atomically with the associated ghost updates.

  We say a state $a\in A$ \emph{satisfies the channel invariant} if all
  buffered messages in $a$ satisfy $\msginv$.  Precisely, for any
  $s,d\in\PID$, let $l$ be the value of $\vrc[s]$ on process $d$ in $a$,
  and let $m$ be the value of $\vsc[d]$ on process $s$ in $a$.  (Note on
  any reachable state, we must have $l\leq m$.)  For any $k$ satisfying
  $l\leq k<m$, the message $x$ at position $k-l$ in queue $s\ra d$ in
  $a$ satisfies $\phi(s,d,k,x)$.  (We assume the oldest message is at
  position $0$ in the queue.)  We will show this invariant is preserved
  on any execution of $A$.

  For $0\leq i<n$, let $B_i$ denote the set of states of $S$
  with the same parameter assignment, $\vNP$ assigned $n$, and $\vpid$
  assigned $i$.  Let $B=B_0\times\cdots\times B_{n-1}$.  We consider $B$
  to be a transition system in which each transition is an atomic step
  from one of the components, with a similar exception: the 4 statements
  simulating a $\kwsend$ or $\kwrecv$ execute in a single atomic step.
  Note the components of $B$ are completely independent: there are no
  message buffers, nor any communication between these components.

  We aim to establish a simulation between the two transition systems
  $A$ and $B$.  There is an obvious mapping from control points in $P$
  to corresponding points in $S$.  Define a binary relation $R$
  on $A\times B$: $aRb$ holds iff for each $i\in0..n-1$, the state of
  process $i$ in $a$ corresponds exactly to that of component $i$ in
  $b$, i.e., the program counters correspond and the values of all
  variables are the same.

  Let $a\stackrel{t}{\ra}a'$ be a transition that occurs on some
  execution in $A$.  Suppose $a$ satisfies the channel invariant and
  $aRb$.  We claim that $a'$ satisfies the channel invariant and there
  is a transition $b{\ra}b'$ in $B$ such that $a'Rb'$.  If $t$ is not a
  communication transition, this is clear, since the values of the
  variables are the same in $a$ and $b$, and $\msginv$ does not depend
  on the state.

  So suppose $t$ is a \kwsend.  The corresponding transition $b\ra b'$
  in $B$ checks assertions (which hold, by the assumption that
  $S$ is correct), and updates $\vsc$ and $\vlevel$ in the same
  way as $t$.  Hence $a'Rb'$.  To see that $a'$ satisfies the channel
  invariant: the old messages in the queue continue to satisfy \msginv,
  and the newly inserted message satisfies \msginv{} since the assertion
  in $S$ passed.

  Suppose $t$ is a \kwrecv.  Let $v$ be the value of the message that
  was received in $A$. In $B$, we are allowed to choose any value
  satisfying $\msginv$.  We choose $v$, which we know satisfies
  $\msginv$ because $a$ satisfies the channel invariant.  This yields a
  transition $b\ra b'$ with $a'Rb'$.  The state $a'$ satisfies the
  channel invariant since we have only removed one message---the
  remaining messages continue to satisfy $\msginv$.

  Now given any execution of $A$, we construct a corresponding execution
  of $B$ by repeatedly using the result above.  Each assertion in $A$
  must hold, because the corresponding assertion in $B$ uses the same
  variable values.  An execution of any component of $B$ is an execution
  of $S$, and we have assumed $S$ is correct.

  The execution of $A$ cannot be infinite, because that would imply at
  least one component of $B$ has an infinite execution, contradicting
  the assumption that $S$ is correct.

  Suppose an execution of $A$ ends without every process terminating.
  If some process in $A$ is blocked at a $\kwchoose$ statement with an
  unsatisfiable condition, the corresponding statement would be blocked
  for the same reason (as the variables have the same values),
  contradicting the correctness of $S$.  Hence $A$ must be
  deadlocked: all non-terminated processes are blocked at receive
  statements for which the buffers are empty.  Each such process is
  waiting to receive a message and obtain a new level, specified by
  $\rlevel$.

  Among all such processes, let $p$ be one with minimal new level.  Say
  $p$ is waiting for message number $k$ from process $s$.  Process $s$
  has already sent $k$ messages to $p$, all of which have been received
  (messages are numbered from $0$). The new level of $p$ will be
  $\rlevel(s,p,k)$, and the new value of $\vrc[s]$ on process $p$ will
  be $k+1$.

  We must have $k<\nummsg(s,p)$ because all executions of $S$ terminate,
  including the execution with $\vpid=p$.  But $\vrc[s]$ only increases,
  and just before termination an assertion is checked which guarantees
  $\vrc[s]=\nummsg(s,p)$.

  Now process $s$ cannot be terminated (else it would have violated its
  final assertion $\vsc[p]=\nummsg(s,p)$).  Hence $s$ is also blocked at
  a receive, say from process $q$, and its new level will be
  $\rlevel(q,s,l)$ for some $l\geq 0$.

  The sequential program with $\vpid=s$ can run to termination, and at
  some point it must simulate send number $k$ to $p$, updating its level
  to $\slevel(s,p,k)$.  Since the level only increases (again, because
  all assertions in program $s$ pass), we have
  \[
    \rlevel(q,s,l) < \slevel(s,p,k).
  \]
  On the other hand, $\slevel(s,p,k)<\rlevel(s,p,k)$ by
  \eqref{eq:specok}.    Hence
  \[
    \rlevel(q,s,l) < \slevel(s,p,k) < \rlevel(s,p,k),
  \]
  contradicting the minimality of $p$'s new level.

  Hence the execution of $P$ does not deadlock, and all processes
  terminate normally, and all assertions hold.  This shows $P$
  is correct.
\end{proof}

\onecolumn

\section{Listings from Case Study}

\subsection{\code{cycsum.c} and \code{Cyclic.h}}

% \inputminted{c}{code/cycsum.c}
\begin{lstlisting}[style=arxivC, language=C]
/* cycsum.c: MPI sum reduction using cyclic nearest neighbor communication
   This program can be verified using Frama-C+WP.  See Makefile.
 */
#include <assert.h>
#include <stdlib.h>
#include <limits.h>
#include <mpi.h>
#include "Cyclic.h"

typedef double T;
#define DT MPI_DOUBLE
typedef unsigned long ulong;
#define COMM MPI_COMM_WORLD
#define STAT MPI_STATUS_IGNORE
#define TAG 1

int nprocs, rank; // number of processes, rank of this process

#define GLOBALMEM VM_Pointers, &rank, &nprocs

/*@
  // A read-only array, disjoint from all other objects...
  predicate Array(T * X, integer N) =
    \valid_read(X+(0..(N)-1)) && \separated(GLOBALMEM, X+(0..(N)-1));

  predicate Init = 2*VM_NP < INT_MAX && VM_Init && \separated(GLOBALMEM);
*/

/* Sums the x's over all procs (collective).
   Ghost parameter A: array of length VM_NP of x values on each process
*/
/*@ mpi collective;
    mpi universal A[0 .. VM_NP-1];
    requires Init && VM_state==VM_Active && rank==VM_pid && nprocs==VM_NP;
    requires Array(A, VM_NP) && x==A[VM_pid];
    ensures E1: \result == Cyclic_Sum(A, VM_NP, 0, VM_NP);
    ensures E2: VM_UnchangedGlobal{Pre,Here};
    assigns VM_GlobalVars;
 */
T sum(T x) /*@ ghost (\ghost const T * A) */ {
  int left = (rank+nprocs-1)%nprocs;
  int right = (rank+1)%nprocs;
  T result = x;
  /*@ mpi begin regions: nregions = 1; region(i) = 1; */
  /*@ mpi begin region 1:
      nummsg(src,dest) = (dest==(src+VM_NP-1)%VM_NP ? VM_NP-1 : 0);
      mcount(src,dest,idx) = 1;
      mdtype(src,dest,idx) = DT;
      msgtag(src,dest,idx) = TAG;
      msginv(src,dest,idx,buf,count,dt) = *(buf)==A[(src+idx)%VM_NP];
      slevel(src,dest,idx) = idx+1;
  */
  /*@
    loop invariant L1: 1<=i<=nprocs;
    loop invariant L2: x == A[(rank+i-1)%nprocs];
    loop invariant L3: result==Cyclic_Sum(A, nprocs, rank, i);
    loop invariant L4: VM_sc[left] == VM_rc[right] == i-1;
    loop invariant L5: \forall integer j; (0<=j<nprocs && j!=left) ==> VM_sc[j]==0;
    loop invariant L6: \forall integer j; (0<=j<nprocs && j!=right) ==> VM_rc[j]==0;
    loop invariant L7: VM_level == i-1;
    loop assigns i, result, x, VM_level, VM_sc[left], VM_rc[right];
    loop variant nprocs-i;
  */
  for (int i=1; i<nprocs; i++) {
    T y;
    //@ ghost L1: ;
    MPI_Sendrecv(&x, 1, DT, left, TAG, &y, 1, DT, right, TAG, COMM, STAT);
    //@ assert A1: (rank+i)%nprocs == (((rank+1)%nprocs) + i-1 )%nprocs;
    //@ ghost Cyclic_Unchanged(L1, Here, A, nprocs, rank, i);
    x = y;
    result += x;
  }
  /*@ assert \forall integer src; 0<=src<nprocs ==> 
      (src==right <==> rank==(src+nprocs-1)%nprocs); */
  //@ ghost Cyclic_lemma2(A, nprocs, rank);
  //@ mpi end region 1;
  //@ mpi end regions;
  return result;
}

/*@ mpi collective;
    mpi universal A[0 .. VM_NP-1];
    requires Init && VM_state==VM_Uninitialized;
    requires Array(A, VM_NP);
    requires \forall integer i; 0<=i<VM_NP ==> A[i] == (T)i;
    ensures \result;
    ensures VM_UnchangedGlobalExceptState{Pre,Here};
    assigns VM_GlobalVars, nprocs, rank;
 */
_Bool test1(void) /*@ ghost (\ghost const T * A) */ {
  MPI_Init(NULL, NULL);
  MPI_Comm_size(MPI_COMM_WORLD, &nprocs);
  MPI_Comm_rank(MPI_COMM_WORLD, &rank);
  T x = (T)rank, expected = (T)0;
  /*@ mpi begin regions:
      nregions = 1;
      region(i) = sum;
      sum#A(i,j) = j;
  */
  T actual = sum(x) /*@ ghost (A) */;
  //@ mpi end regions;
  MPI_Finalize();
  /*@ loop invariant 0<=i<=nprocs && expected == Cyclic_Sum(A, nprocs, 0, i);
      loop assigns i, expected;
      loop variant nprocs - i;
   */
  for (int i=0; i<nprocs; i++) expected += (T)i;
  return actual == expected;
}

int main(void) {
  assert(test1() /*@ ghost (NULL) */);
}
\end{lstlisting}

% \inputminted{c}{code/Cyclic.h}
\begin{lstlisting}[style=arxivC, language=C]
/* Cyclic.h : logical theory of cyclic sums
 */
typedef double T;

/*@
  logic T Cyclic_Sum(T * a, integer length, integer start, integer n) =
    n <= 0 ? (T)0 :
    (T)(Cyclic_Sum(a, length, start, n-1) + a[(start+n-1)%length]);

  lemma Cyclic_Init:
    \forall T * a, integer length, integer start, integer n;
      n <= 0 ==> Cyclic_Sum(a, length, start, n) == 0;

  lemma Cyclic_Next:
    \forall T * a, integer length, integer start, integer n;
      (0<=n<length && 0<=start<length) ==>
      Cyclic_Sum(a, length, start, n+1) ==
        Cyclic_Sum(a, length, start, n) + a[(start+n)%length];

  lemma Cyclic_arith: \forall integer pid,i,NP; 0 <= pid < NP && 0 < i ==>
    ((pid+1)%NP+i)%NP ==  (pid+i+1)%NP;
*/

// "Lemma functions" which can be invoked in ghost code to prove a fact...
/*@ ghost
  /@
    requires 1<=n<=length;
    requires 0<=start<length;
    requires \valid_read(a+(0..n-1));
    ensures Cyclic_Sum(a, length, start, n) ==
      a[start] + Cyclic_Sum(a, length, start+1, n-1);
    assigns \nothing;
  @/
  void Cyclic_lemma1(T const \ghost * a, int length, int start, int n) {
    /@
      loop invariant L0: 1<=i<=n;
      loop invariant L1: Cyclic_Sum(a,length,start,i) == 
        a[start] + Cyclic_Sum(a,length,start+1,i-1);
      loop assigns i;
      loop variant n-i;
    @/
    for (int i=1; i<n; i++);
  }

  /@
    requires 0<=start<n;
    requires \valid_read(a+(0..n-1));
    ensures Cyclic_Sum(a, n, start, n) == Cyclic_Sum(a, n, 0, n);
    assigns \nothing;
  @/
  void Cyclic_lemma2(T const \ghost * a, int n, int start) {
    /@
      loop invariant L0: 0<=i<=start;
      loop invariant L1: Cyclic_Sum(a,n,i,n) == Cyclic_Sum(a,n,0,n);
      loop assigns i;
      loop variant start-i;
    @/  
    for (int i=0; i<start; i++) {
      Cyclic_lemma1(a, n, i, n);
      /@ assert (i+n)%n==i; @/
    }
  }
*/

/* A "lemma macro" to prove that Cyclic_Sum(a, length, start, n)
   doesn't change if the arguments don't */
#define Cyclic_Unchanged(L1, L2, a, length, start, n)                   \
  /@                                                                    \
    loop invariant SU1: 0 <= _i <= n;                                   \
    loop invariant SU2: \forall integer _j; _j==_i ==>                  \
      \at(Cyclic_Sum(a,length,start,_j),L1) ==                          \
      \at(Cyclic_Sum(a,length,start,_j),L2);                            \
    loop assigns _i;                                                    \
    loop variant n - _i;                                                \
  @/                                                                    \
  for (int _i=0; _i<n; _i++);
\end{lstlisting}

\subsection{\code{allsum.c}}

% \inputminted{c}{code/allsum.c}
\begin{lstlisting}[style=arxivC, language=C]
#include <assert.h>
#include <stdlib.h>
#include <limits.h>
#include <mpi.h>

typedef double T;
#define DT MPI_DOUBLE
typedef unsigned long ulong;
#define COMM MPI_COMM_WORLD
#define STAT MPI_STATUS_IGNORE
#define IN_TAG 1
#define OUT_TAG 2

int nprocs, rank; // number of processes, rank of this process

#define GLOBALMEM VM_Pointers, &rank, &nprocs

/*@
  // A read-only array, disjoint from all other objects...
  predicate Array(T * X, integer N) =
    \valid_read(X+(0..(N)-1)) && \separated(GLOBALMEM, X+(0..(N)-1));

  predicate Init = VM_Init && \separated(GLOBALMEM);

  logic T Sum(T * a, integer n) = n <= 0 ? (T)0 : (T)(Sum(a, n-1) + a[n-1]);

  lemma Sum_Init:
    \forall T * a, integer n; n <= 0 ==> Sum(a, n) == 0;

  lemma Sum_Next:
    \forall T * a, integer n; n >= 0 ==> Sum(a, n+1) == (T)(Sum(a, n) + a[n]);
*/

// A "lemma macro" to prove that Sum(X,N) doesn't change if X and N don't:
#define Sum_Unchanged(L1, L2, x, n)                             \
  /@ loop invariant 0 <= _i <= n;                               \
     loop invariant \forall integer _j; _j==_i ==>              \
       \at(Sum(x,_j),L1) == \at(Sum(x,_j),L2);                  \
     loop assigns _i;                                           \
     loop variant n - _i;                                       \
   @/                                                           \
  for (int _i=0; _i<n; _i++);

/* Sums the x's over all procs (collective).
   Ghost parameter X: array of length VM_NP of x values on each process
*/
/*@ mpi collective;
    mpi universal X[0 .. VM_NP-1];
    requires Init && VM_state==VM_Active && rank==VM_pid && nprocs==VM_NP;
    requires Array(X, VM_NP) && x==X[VM_pid];
    ensures \result == Sum(X, VM_NP);
    ensures VM_UnchangedGlobal{Pre,Here};
    assigns VM_GlobalVars;
 */
T sum(T x) /*@ ghost (\ghost const T * X) */ {
  T result;
  /*@ mpi begin regions: nregions = 1; region(i) = 1; */
  /*@ mpi begin region 1:
      nummsg(src,dest) = (src==0 ? (dest==0 ? 0 : 1) : (dest==0 ? 1 : 0));
      mcount(src,dest,idx) = 1;
      mdtype(src,dest,idx) = DT;
      msgtag(src,dest,idx) = (src==0 ? OUT_TAG : IN_TAG);
      msginv(src,dest,idx,buf,count,dt) = 
        *(buf)==(src==0 ? Sum(X, VM_NP) : X[src]);
      slevel(src,dest,idx) = (src==0 ? (ulong)VM_NP+dest-1 : src);
  */
  if (rank == 0) {
    T buf;
    result = x;
    /*@ loop invariant 1<=i<=VM_NP;
        loop invariant result==Sum(X, i);
        loop invariant \forall integer j; 1<=j<i ==> VM_rc[j]==1;
        loop invariant \forall integer j; i<=j<VM_NP ==> VM_rc[j]==0;
        loop invariant VM_level == i-1;
        loop assigns i, result, buf, VM_level, VM_rc[1..VM_NP-1];
        loop variant VM_NP-i;
     */
    for (int i=1; i<nprocs; i++) {
      //@ ghost L1: ;
      MPI_Recv(&buf, 1, DT, i, IN_TAG, COMM, STAT);
      result += buf;
      //@ ghost Sum_Unchanged(L1, Here, X, i+1)
    }
    /*@ loop invariant 1<=i<=VM_NP;
        loop invariant \forall integer j; 1<=j<i ==> VM_sc[j]==1;
        loop invariant \forall integer j; i<=j<VM_NP ==> VM_sc[j]==0;
        loop invariant VM_level == VM_NP+i-2;
        loop assigns i, VM_level, VM_sc[1..VM_NP-1];
        loop variant VM_NP-i;
     */
    for (int i=1; i<nprocs; i++)
      MPI_Send(&result, 1, DT, i, OUT_TAG, COMM);
  } else {
    MPI_Send(&x, 1, DT, 0, IN_TAG, COMM);
    MPI_Recv(&result, 1, DT, 0, OUT_TAG, COMM, STAT);
  }
  //@ mpi end region 1;
  //@ mpi end regions;
  return result;
}

/*@ mpi collective;
    mpi universal X[0 .. VM_NP-1];
    requires Init && VM_state==VM_Uninitialized;
    requires Array(X, VM_NP);
    requires \forall integer i; 0<=i<VM_NP ==> X[i] == (T)i;
    ensures \result;
    ensures VM_UnchangedGlobalExceptState{Pre,Here};
    assigns VM_GlobalVars, nprocs, rank;
 */
_Bool test1(void) /*@ ghost (\ghost const T * X) */ {
  MPI_Init(NULL, NULL);
  MPI_Comm_size(MPI_COMM_WORLD, &nprocs);
  MPI_Comm_rank(MPI_COMM_WORLD, &rank);
  T x = (T)rank, expected = (T)0;
  /*@ mpi begin regions:
      nregions = 1;
      region(i) = sum;
      sum#X(i,j) = j;
  */
  T actual = sum(x) /*@ ghost (X) */;
  //@ mpi end regions;
  MPI_Finalize();
  /*@ loop invariant 0<=i<=VM_NP && expected == Sum(X,i);
      loop assigns i, expected;
      loop variant VM_NP-i;
   */
  for (int i=0; i<nprocs; i++) expected += (T)i;
  return actual==expected;
}

int main(void) {
  assert(test1() /*@ ghost (NULL) */);
}
\end{lstlisting}

\subsection{\code{diffuse1d.c}}

% \inputminted{c}{code/diffuse1d.c}
\begin{lstlisting}[style=arxivC, language=C]
#include <assert.h>
#include <limits.h>
#include <mpi.h>
#include <stdlib.h>

#define FIRST(pid,nx,np) ((pid)*(nx)/(np))
#define NXL(pid,nx,np) (FIRST((pid)+1,nx,np) - FIRST(pid,nx,np))
#define COMM MPI_COMM_WORLD
#define STAT MPI_STATUS_IGNORE
#define TAG 1

int nprocs, rank; // number of processes, rank of this process
unsigned int NSTEP, nxl;
unsigned long first, NX;
double K, * u, * u_new;

/*@
  axiomatic D {
    logic double diffuse0(integer i);
    axiom DF{L1,L2}: \forall integer i;
      \at(diffuse0(i), L1) == \at(diffuse0(i), L2);
  }

  logic double diffuse1(double c, double left, double x, double right) = 
    (double)(x+c*(right+left-2.0*x));

  logic double diffuse(integer time, integer nx, double c, integer i)
  = time <= 0 ? diffuse0(i) :
    i == 0 ? diffuse0(0) :
    i == nx+1 ? diffuse0(nx+1) :
    diffuse1(c, diffuse(time-1, nx, c, i-1),
                diffuse(time-1, nx, c, i),
                diffuse(time-1, nx, c, i+1));

  // This obvious lemma speeds things up...
  lemma Diffuse0: \forall integer nx, double c, integer i;
    diffuse(0, nx, c, i) == diffuse0(i);

  lemma DiffuseBoundLeft: \forall integer time, integer nx, double c;
    diffuse(time, nx, c, 0) == diffuse0(0);

  lemma DiffuseBoundRight: \forall integer time, integer nx, double c;
    diffuse(time, nx, c, nx+1) == diffuse0(nx+1);

  predicate isDiffuse(integer time, integer f, double * p,
    integer lo, integer hi)
  = \forall integer i; lo<=i<=hi ==> p[i] == diffuse(time, NX, K, f+i);

  predicate Boundaries(double * p) =
    (VM_pid==0 ==> p[0]==diffuse0(0)) &&
    (VM_pid==VM_NP-1 ==> p[nxl+1]==diffuse0(NX+1));

  predicate Init(integer n) = VM_Init &&
    \valid(u+(0..n-1)) && \valid(u_new+(0..n-1)) &&
    \separated(VM_Pointers, &rank, &nprocs, &NSTEP, &nxl, &first,
      &NX, &K, u+(0..n-1), u_new+(0..n-1));
*/

/* Computes one time step.  Given u[0..nxl+1], computes
   u_new[1..nxl]. */
/*@ requires Init(nxl+2) && nxl < UINT_MAX && first <= NX - nxl;
    requires isDiffuse(time, first, u, 0, nxl+1);
    ensures  isDiffuse(time+1, first, u_new, 1, nxl);
    assigns  u_new[1..nxl];
*/
void update(void) /*@ ghost (unsigned int time) */ {
  /*@ loop invariant 1 <= i <= nxl+1;
      loop invariant \forall integer j; 1 <= j < i ==>
        u_new[j] == diffuse1(K, u[j-1], u[j], u[j+1]);
      loop assigns i, u_new[1..nxl];
      loop variant nxl - i; */
  for (unsigned int i=1; i<=nxl; i++)
    u_new[i] = u[i]+K*(u[i+1]+u[i-1]-2.0*u[i]);
}

/* Performs one ghost cell exchange, in two parts.  In part 1,
   everyone sends to their right and receives from their left.  In
   part 2, everyone sends left and receives from right.  Exceptions
   are made for the first and last processes. */
/*@ mpi collective;
    mpi universal time, NX, K;
    requires R1: Init(nxl+2) && VM_state == VM_Active;
    requires R2: rank == VM_pid && nprocs == VM_NP <= NX;
    requires R3: first == FIRST(rank,NX,nprocs);
    requires R4: nxl == NXL(rank,NX,nprocs) < UINT_MAX;
    requires D:  isDiffuse(time, first, u, 1, nxl);
    requires B:  Boundaries(u);
    ensures  D:  isDiffuse(time, first, u, 0, nxl+1);
    ensures  B:  Boundaries(u);
    ensures  U:  VM_UnchangedGlobal{Pre,Here};
    assigns      VM_GlobalVars, u[0], u[nxl+1];
 */
void exchange(void) /*@ ghost (unsigned int time) */ {
  //@ assert A1: VM_pid==VM_NP-1 ==> u[nxl+1]==diffuse(time, NX, K, first+nxl+1);
  int left = rank > 0 ? rank - 1 : MPI_PROC_NULL;
  int right = rank < nprocs - 1 ? rank + 1 : MPI_PROC_NULL;

  /*@ mpi begin regions: nregions = 2; region(i) = i+1; */
  /*@ mpi begin region 1:
      nummsg(s,d) = ((s<VM_NP-1 && d==s+1) ? 1 : 0);
      mcount(s,d,idx) = 1;
      mdtype(s,d,idx) = MPI_DOUBLE;
      msgtag(s,d,idx) = TAG;
      msginv(s,d,idx,buf,count,dt) = *(buf) == diffuse(time, NX, K, (d)*NX/VM_NP);
      slevel(s,d,idx) = 1;
  */
  MPI_Sendrecv(u+nxl, 1, MPI_DOUBLE, right, TAG,
	       u, 1, MPI_DOUBLE, left, TAG, COMM, STAT);
  //@ mpi end region 1;
  /*@ mpi begin region 2:
      nummsg(s,d) = ((s>0 && d==s-1) ? 1 : 0);
      mcount(s,d,idx) = 1;
      mdtype(s,d,idx) = MPI_DOUBLE;
      msgtag(s,d,idx) = TAG;
      msginv(s,d,idx,buf,count,dt) =
        *(buf) == diffuse(time, NX, K, FIRST(s,NX,VM_NP)+1);
      slevel(s,d,idx) = 1;
  */
  MPI_Sendrecv(u+1, 1, MPI_DOUBLE, left, TAG,
	       u+nxl+1, 1, MPI_DOUBLE, right, TAG, COMM, STAT);
  //@ mpi end region 2;
  //@ mpi end regions;
  //@ assert D1: isDiffuse(time, first, u, 1, nxl);
  //@ assert D2: u[0] == diffuse(time, NX, K, first);
  //@ assert D3: u[nxl+1] == diffuse(time, NX, K, first+nxl+1);
  //@ assert D4: \forall integer i; 0<=i<=nxl+1 ==> u[i] == diffuse(time, NX, K, first+i);
}

/* The initial values of diffuse.  Do not verify. */
/*@ requires nx>=0 && 0<=i<=nx+1;
    ensures \result == diffuse0(i);
    assigns \nothing;
*/    
double initialCond(unsigned long nx, unsigned long i) {
  return (i == 0 || i == nx+1)  ? 0.0 : 100.0;
}

/* Runs the simulation */
/*@ mpi collective;
    mpi universal NSTEP, NX, K;
    requires R1:  Init(nxl+2);
    requires R2:  VM_state == VM_Active;
    requires R3:  rank == VM_pid && nprocs == VM_NP <= NX < ULONG_MAX;
    requires R4:  nxl == NXL(rank, NX, nprocs) < UINT_MAX;
    requires R5:  first == FIRST(rank, NX, nprocs);
    ensures  E1:  isDiffuse(NSTEP, first, u, 1, nxl);
    ensures  E2:  VM_UnchangedGlobal{Pre,Here};
    assigns       VM_GlobalVars, u, u_new, u[0..nxl+1], u_new[0..nxl+1];
 */
void run(void) {
  //@ assert HARD: first + nxl <= NX;
  if (rank == 0)
    u[0] = u_new[0] = initialCond(NX, 0);
  if (rank == nprocs - 1)
    u[nxl+1] = u_new[nxl+1] = initialCond(NX, NX+1);
  /*@ loop invariant I1: 1 <= i <= nxl+1;
      loop invariant I2: isDiffuse(0, first, u, 1, i-1);
      loop assigns       i, u[1..nxl];
      loop variant       nxl-i;
  */
  for (unsigned int i = 1; i <= nxl; i++)
    u[i] = initialCond(NX, first+i);

  /*@ mpi begin regions:
      nregions = NSTEP;
      region(i) = exchange;
      exchange#time(i) = i;
      exchange#NX(i) = NX;
      exchange#K(i) = K;
  */
  //@ ghost double * u1 = u, * u2 = u_new;
  /*@ loop invariant I3: Init(nxl+2);
      loop invariant I4: time <= NSTEP;
      loop invariant I5: isDiffuse(time, first, u, 1, nxl);
      loop invariant I6: Boundaries(u) && Boundaries(u_new);
      loop invariant I7: (u==u1 && u_new==u2) || (u==u2 && u_new==u1);
      loop invariant I8: VM_UnchangedGlobal{Pre,Here};
      loop invariant I9: VM_regionCount == time;
      loop assigns       VM_GlobalVars, VM_regionCount, time, u, u_new,
                         u1[0..nxl+1], u2[0..nxl+1];
      loop variant       NSTEP - time;  */
  for (unsigned int time = 0; time < NSTEP; time++) {
    exchange() /*@ ghost(time) */;
    update() /*@ ghost(time) */;
    double * tmp = u; u = u_new; u_new = tmp;
  }
  //@ mpi end regions;
}

/* Runs a smallish test. Do not verify. */
int main(void) {
  NX = 100, K = 1.0/100.0, NSTEP = 10;
  MPI_Init(NULL, NULL);
  MPI_Comm_size(MPI_COMM_WORLD, &nprocs);
  MPI_Comm_rank(MPI_COMM_WORLD, &rank);
  first = FIRST(rank, NX, nprocs);
  nxl = NXL(rank, NX, nprocs);
  u = malloc((nxl+2)*sizeof(double));
  u_new = malloc((nxl+2)*sizeof(double));
  run();
  MPI_Finalize();
  free(u);
  free(u_new);
}
\end{lstlisting}

\subsection{\code{bcast.c}}

% \inputminted{c}{code/bcast.c}
\begin{lstlisting}[style=arxivC, language=C]
#include <assert.h>
#include <mpi.h>
#include <stdlib.h>

typedef double T;
#define DT MPI_DOUBLE

/*@
  predicate Equal(T * p, T * q, integer n) =
  \forall integer i; 0<=i<n ==> p[i]==q[i];
*/

/*@ ghost
  /@ requires n>=0 && \valid(p+(0..n-1));
     terminates \true;
     ensures \forall integer i; 0<=i<n ==> p[i]==(T)i;
     assigns p[0..n-1];
     exits \false;
   @/
  void setLinear(\ghost T * p, int n);
*/

/*@ mpi collective;
    mpi universal count, datatype, root, Data[0..count-1];
    requires count >= 0;
    requires comm == MPI_COMM_WORLD;
    requires VM_Init && VM_state == VM_Active;
    requires 0 <= root < VM_NP;
    requires \valid(buf + (0 .. count-1));
    requires \valid_read(Data + (0 .. count-1));
    requires \separated(VM_Pointers, buf+(0 .. count-1));
    requires VM_pid == root ==> Equal(buf, Data, count);
    ensures Equal(buf, Data, count);
    ensures VM_UnchangedGlobal{Pre,Here};
    assigns buf[0 .. count-1], VM_GlobalVars;
 */
int Bcast(T * buf, int count, 
	  MPI_Datatype datatype, int root, MPI_Comm comm)
/*@ ghost ( \ghost const T * const Data ) */
{
  int nprocs, rank;
  int tag = 999;

  MPI_Comm_size(comm, &nprocs);
  MPI_Comm_rank(comm, &rank);
  /*@ mpi begin regions: nregions = 1; region(i) = 1; */
  /*@ mpi begin region 1:
      nummsg(src,dest) = ((src==root && src!=dest) ? 1 : 0);
      mcount(src,dest,idx) = count;
      mdtype(src,dest,idx) = datatype;
      msgtag(src,dest,idx) = 999;
      msginv(src,dest,idx,buf,count,dt) = Equal(buf, Data, count);
      slevel(src,dest,idx) = dest+1;
   */
  if (rank == root) {
    /*@ loop invariant 0<=i<=nprocs;
        loop invariant nprocs == VM_NP;
	loop invariant Equal(buf, Data, count);
	loop invariant VM_level == (i==root+1 ? root : i);
	loop invariant \forall integer j; 0<=j<i ==>
	  VM_sc[j] == (j==root ? 0 : 1);
	loop invariant \forall integer j; i<=j<VM_NP ==>
	  VM_sc[j] == 0;
	loop assigns i, VM_level, VM_sc[0..VM_NP-1];
	loop variant nprocs - i;
     */
    for (int i = 0; i < nprocs; i++)
      if (i != root)
        MPI_Send(buf, count, datatype, i, tag, comm);
  } else
    MPI_Recv(buf, count, datatype, root, tag, comm,
             MPI_STATUS_IGNORE);
  //@ mpi end region 1;
  //@ mpi end regions;
  return 0;
}

/*@ mpi collective;
    mpi universal n;
    requires VM_Init && VM_state == VM_Active;
    requires n >= 0;
    requires \valid(buf + (0 .. n-1));
    requires \valid(Data + (0 .. n-1));
    requires \separated(VM_Pointers, buf+(0..(n-1)), Data+(0..(n-1)));
    ensures \forall integer i; 0<=i<n ==> buf[i] == (T)i;
    ensures VM_UnchangedGlobal{Pre,Here};
    assigns VM_GlobalVars, buf[0..n-1], Data[0..n-1];
 */
void test(int n, T * buf) /*@ ghost (\ghost T * Data) */ {
  int rank, nprocs;
  MPI_Comm_size(MPI_COMM_WORLD, &nprocs);
  MPI_Comm_rank(MPI_COMM_WORLD, &rank);
  /*@ ghost setLinear(Data, n); */
  if (rank == 0) {
    /*@ loop invariant L1: 0<=i<=n;
        loop invariant L2: \forall integer j; 0<=j<i ==> buf[j]==(T)j;
	loop assigns i, buf[0 .. n-1];
	loop variant n-i;
    */
    for (int i=0; i<n; i++)
      buf[i] = (T)i;
  }
  /*@ mpi begin regions:
      nregions = 1;
      region(i) = Bcast;
      Bcast#count(i) = n;
      Bcast#datatype(i) = DT;
      Bcast#root(i) = 0;
      Bcast#Data(i,j) = (T)j;
  */
  Bcast(buf, n, DT, 0, MPI_COMM_WORLD) /*@ ghost (Data) */;
  //@ mpi end regions;
}

int main(void) {
  int rank, n=10;
  T * buf = malloc(n*sizeof(T));
  MPI_Init(NULL, NULL);
  MPI_Comm_rank(MPI_COMM_WORLD, &rank);
  if (rank == 0) {
    for (int i=0; i<n; i++)
      buf[i] = (T)i;
  }
  // Data is irrelevant because we're not verifying this main function
  test(n, buf) /*@ ghost (NULL) */;
  MPI_Finalize();
  for (int i=0; i<n; i++)
    assert(buf[i] == (T)i);
  free(buf);
}
\end{lstlisting}

\subsection{\code{gather.c}}

% \inputminted{c}{code/gather.c}
\begin{lstlisting}[style=arxivC, language=C]
#include <assert.h>
#include <limits.h>
#include <mpi.h>
#include <stdlib.h>

typedef double T;
#define DT MPI_DOUBLE

/*@
  predicate Equal(T * p, T * q, integer lo, integer hi) =
    \forall integer i; lo<=i<hi ==> p[i]==q[i];

  predicate Unchanged{K,L}(T * p, integer lo, integer hi) =
    \forall integer i; lo<=i<hi ==> \at(p[i],K) == \at(p[i],L);

  lemma EqUnchanged{K,L}:
    \forall T * p, T * q, integer lo, hi;
    (Equal{K}(p, q, lo, hi) && Unchanged{K,L}(p, lo, hi) &&
     Unchanged{K,L}(q, lo, hi)) ==>
    Equal{L}(p, q, lo, hi);

  lemma EqConcat:
    \forall T * p, T * q, integer a,b,c;
    (0<=a<=b<=c && Equal(p, q, a, b) && Equal(p, q, b, c))
      ==> Equal(p, q, a, c);
*/

/* The point of the following lemma functions is to say that if
   p+offset and q+offset are equal from lo to hi, then p and q are
   equal from lo+offset to hi+offset. */

/*@ ghost

/@ requires  0<=lo<=i<hi<=INT_MAX;
   requires \valid_read(p+(lo..(hi-1)));
   requires \valid_read(q+(lo..(hi-1)));
   requires Equal(p, q, lo, hi);
   ensures p[i] == q[i];
   assigns \nothing;
@/
void Arrays_Eq(T * p, T const \ghost * q, int lo, int hi, int i) {
}

/@ requires 0<=lo<=hi;
   requires offset >= 0;
   requires hi+offset <= INT_MAX;
   requires \valid_read(p+((offset+lo) .. (offset+hi-1)));
   requires \valid_read(q+((offset+lo) .. (offset+hi-1)));
   requires Equal(p+offset, q+offset, lo, hi);
   ensures Equal(p, q, lo+offset, hi+offset);
   assigns \nothing;
 @/
void Arrays_EqShift(T * p, T const \ghost * q, int lo, int hi, int offset) {
  /@ loop invariant lo <= i <= hi;
     loop invariant \forall integer k; lo+offset <= k < i+offset ==> p[k] == q[k];
     loop assigns i;
     loop variant hi-i;
   @/
  for (int i=lo; i<hi; i++) {
    Arrays_Eq(p+offset, q+offset, lo, hi, i);
  }
}

*/

/*@ ghost
  /@ requires n>=0 && \valid(p+(0..n-1));
     terminates \true;
     ensures \forall integer i; 0<=i<n ==> p[i]==(T)i;
     assigns p[0..n-1];
     exits \false;
   @/
  void setLinear(\ghost T * p, int n);
*/

/*@ mpi collective;
    mpi universal sendcount, sendtype, root, Data[0..(sendcount*VM_NP-1)];
    requires comm == MPI_COMM_WORLD && VM_Init && VM_state == VM_Active;
    requires 0 <= root < VM_NP;
    requires sendcount >= 0;
    requires sendcount * VM_NP <= INT_MAX;
    requires \valid_read(sendbuf + (0 .. sendcount-1));
    requires \valid_read(Data + (0 .. (sendcount*VM_NP-1)));
    requires VM_pid == root ==> recvcount == sendcount;
    requires VM_pid == root ==> recvtype == sendtype;
    requires VM_pid == root ==> \valid(recvbuf + (0 .. (sendcount*VM_NP-1)));
    requires \separated(VM_Pointers, sendbuf+(0 .. sendcount-1),
      recvbuf+(0 .. (sendcount*VM_NP-1)), Data+(0 .. (sendcount*VM_NP-1)));
    requires Equal(sendbuf, Data+sendcount*VM_pid, 0, sendcount);
    ensures VM_pid == root ==> Equal(recvbuf, Data, 0, sendcount*VM_NP);
    ensures VM_UnchangedGlobal{Pre,Here};
    assigns recvbuf[0..(sendcount*VM_NP-1)], VM_GlobalVars;
 */
int Gather(T * sendbuf, int sendcount, MPI_Datatype sendtype,
	   T * recvbuf, int recvcount, MPI_Datatype recvtype,
	   int root, MPI_Comm comm)
/*@ ghost ( \ghost const T * const Data ) */
{
  int nprocs, rank, tag = 998;
  MPI_Comm_size(comm, &nprocs);
  MPI_Comm_rank(comm, &rank);
  /*@ mpi begin regions: nregions = 1; region(i) = 1; */
  /*@ mpi begin region 1:
      nummsg(src,dest) = ((dest==root && src!=dest) ? 1 : 0);
      mcount(src,dest,idx) = sendcount;
      mdtype(src,dest,idx) = sendtype;
      msgtag(src,dest,idx) = 998;
      msginv(src,dest,idx,buf,count,dt) = Equal(buf, Data+src*count, 0, count);
      slevel(src,dest,idx) = src+1;
   */
  if (rank == root) {
    /*@ loop invariant I1: 0<=i<=nprocs;
	loop invariant I2: Unchanged{Pre,Here}(Data+sendcount*VM_pid, 0, sendcount);
	loop invariant I3: Equal(recvbuf, Data, 0, i*recvcount);
	loop invariant I4: VM_level == (i==root+1 ? root : i);
	loop invariant I5: \forall integer j; 0<=j<i ==> VM_rc[j] == (j==root ? 0 : 1);
	loop invariant I6: \forall integer j; i<=j<VM_NP ==> VM_rc[j]==0;
	loop invariant I7: \forall integer j; 0<=j<VM_NP ==> VM_sc[j] == 0;
	loop assigns i, VM_level, VM_rc[0 .. VM_NP-1], recvbuf[0 .. (sendcount*VM_NP-1)];
	loop variant nprocs - i;
    */
    for (int i = 0; i < nprocs; i++)
      if (i == root) {
	/*@ loop invariant J1: 0 <= j <= recvcount == sendcount;
	    loop invariant J2: Unchanged{Pre,Here}(Data + recvcount*rank, 0, recvcount);
	    loop invariant J3: Equal(sendbuf, Data + recvcount*rank, 0, recvcount);
	    loop invariant J4: Equal(recvbuf, Data, 0, i*recvcount + j);
	    loop assigns j, recvbuf[(i*recvcount) .. (recvcount*(i+1) - 1)];
            loop variant recvcount - j;
	 */
	for (int j = 0; j < recvcount; j++) {
	  //@ assert A1: recvcount*i + recvcount - 1 <= recvcount*VM_NP - 1;
	  recvbuf[i*recvcount + j] = sendbuf[j];
	  //@ assert A2: recvbuf[i*recvcount + j] == Data[i*recvcount + j];
	  //@ assert A3: Equal(recvbuf, Data, 0, i*recvcount + j + 1);
	}
      } else {
	//@ assert B1: i*recvcount + recvcount - 1 <= nprocs*sendcount - 1;
        MPI_Recv(recvbuf+i*recvcount, recvcount, recvtype, i, tag, comm,
		 MPI_STATUS_IGNORE);
	//@ assert B2: sendcount*VM_pid + sendcount <= nprocs*sendcount;
	//@ ghost Arrays_EqShift(recvbuf, Data, 0, recvcount, i*recvcount);
      }
  } else {
    MPI_Send(sendbuf, sendcount, sendtype, root, tag, comm);
  }
  //@ mpi end region 1;
  //@ mpi end regions;
  return 0;
}

/*@ mpi collective;
    mpi universal n;
    requires VM_Init && VM_state == VM_Active;
    requires n>=0 && \valid(sbuf + (0 .. n-1));
    requires n * VM_NP <= INT_MAX;
    requires VM_pid == 0 ==> \valid(rbuf + (0 .. (n*VM_NP - 1)));
    requires \valid(Data + (0 .. (n*VM_NP - 1)));
    requires \separated(VM_Pointers, sbuf + (0..(n-1)),
      rbuf + (0..(n*VM_NP-1)), Data+(0..(n*VM_NP-1)));
    ensures VM_pid==0 ==> \forall integer i; 0<=i<VM_NP*n ==> rbuf[i] == (T)i;
    ensures VM_UnchangedGlobal{Pre,Here};
    assigns sbuf[0 .. (n-1)], rbuf[0 .. (n*VM_NP - 1)],
      Data[0 .. (n*VM_NP-1)], VM_GlobalVars;
 */
void test(int n, T * sbuf, T * rbuf) /*@ ghost (\ghost T * Data) */ {
  int rank, nprocs;
  MPI_Comm_size(MPI_COMM_WORLD, &nprocs);
  MPI_Comm_rank(MPI_COMM_WORLD, &rank);
  /*@ ghost setLinear(Data, nprocs*n); */
  /*@ loop invariant 0<=i<=n;
      loop invariant \forall integer j; 0<=j<i ==> sbuf[j]==(T)(rank*n+j);
      loop invariant Equal(sbuf, Data+n*rank, 0, i);
      loop assigns i, sbuf[0 .. n-1];
      loop variant n-i;
  */
  for (int i=0; i<n; i++) {
    //@ assert OF2: 0 <= rank*n <= (VM_NP-1)*n;
    sbuf[i] = (T)(rank*n + i);
  }
  /*@ mpi begin regions:
      nregions = 1;
      region(i) = Gather;
      Gather#sendcount(i) = n;
      Gather#sendtype(i) = DT;
      Gather#root(i) = 0;
      Gather#Data(i,j) = j;
  */
  Gather(sbuf, n, DT, rbuf, n, DT, 0, MPI_COMM_WORLD) /*@ ghost (Data) */;
  //@ mpi end regions;
}

int main(void) {
  MPI_Init(NULL, NULL);
  int rank, nprocs, n = 5;
  MPI_Comm_size(MPI_COMM_WORLD, &nprocs);
  MPI_Comm_rank(MPI_COMM_WORLD, &rank);
  T * sbuf = malloc(n*sizeof(T));
  T * rbuf = rank == 0 ? malloc(n*nprocs*sizeof(T)) : NULL;
  test(n, sbuf, rbuf) /*@ ghost (NULL) */;
  MPI_Finalize();
  if (rank == 0) {
    for (int i=0; i<n*nprocs; i++)
      assert(rbuf[i] == (T)i);
  }
  free(sbuf);
  if (rank == 0) free(rbuf);
}
\end{lstlisting}

\subsection{\code{mpi.h}}

% \inputminted{c}{code/mpi.h}
\begin{lstlisting}[style=arxivC, language=C]
#ifndef _VM_MPI
#define _VM_MPI
/* Verified MPI: MPI header file for Frama-C+WP verification */
#include <stdlib.h>

typedef unsigned long ulong;

#define VM_NO_REGION 0
#define VM_MAX(X,Y) ((X)>=(Y)?(X):(Y))
#define VM_CONCAT(X,Y) X##Y
#define VM_MAC1(X,fid) X##_##fid
#define VM_MAC2(X,fid,rid) X##_##fid##_##rid
#define VM_EXPECTED(fid,i) VM_MAC1(VM_expectedRegion,fid)(i)
#define VM_NUMMSG(fid,rid,s,d) VM_MAC2(VM_nummsg,fid,rid)(s,d)
#define VM_MCOUNT(fid,rid,s,d,j) VM_MAC2(VM_mcount,fid,rid)(s,d,j)
#define VM_MDTYPE(fid,rid,s,d,j) VM_MAC2(VM_mdtype,fid,rid)(s,d,j)
#define VM_MSGTAG(fid,rid,s,d,j) VM_MAC2(VM_msgtag,fid,rid)(s,d,j)
#define VM_MSGINV(fid,rid,s,d,j,buf,count,dt) \
  VM_MAC2(VM_msginv,fid,rid)(s,d,j,buf,count,dt)
#define VM_SLEVEL(fid,rid,s,d,j) VM_MAC2(VM_slevel,fid,rid)(s,d,j)
#define VM_HAVOCA(fid,rid,buf) VM_MAC2(VM_havoca,fid,rid)(buf)
typedef int MPI_Datatype;
#define MPI_DOUBLE 100
typedef int MPI_Comm;
#define MPI_COMM_WORLD 999
#define MPI_PROC_NULL -1
#define MPI_ANY_TAG -1
typedef int MPI_Status;
#define MPI_STATUS_IGNORE (MPI_Status*)0
// Note: the following does not include constants VM_NP and VM_pid:
#define VM_GlobalVars VM_state, VM_done, VM_sc[0..VM_NP-1], VM_rc[0..VM_NP-1]
#define VM_LocalVars VM_currentRegion, VM_regionCount, VM_nregions, VM_level
#define VM_Vars VM_GlobalVars, VM_LocalVars
// The following includes pointers to all global memory in VM:
#define VM_Pointers VM_sc+(0..VM_NP-1), VM_rc+(0..VM_NP-1), &VM_NP, \
    &VM_pid, &VM_state, &VM_done
#define VM_Init 0<=VM_pid<VM_NP &&                              \
    \valid(VM_sc+(0..VM_NP-1)) && \valid(VM_rc+(0..VM_NP-1)) && \
    (\forall integer i; 0<=i<VM_NP ==> VM_sc[i]==VM_rc[i]==0)

// Global ghost code included once for all programs...
/*@ ghost

  typedef enum VM_State {
    VM_Uninitialized,
    VM_Active,
    VM_Finalized
  } VM_State;

  VM_State VM_state = VM_Uninitialized;
  _Bool VM_done; // flag set when regions complete
  int VM_NP, VM_pid; // number of processes, ID number of this process
  unsigned int \ghost * VM_sc, * VM_rc;

  /@ requires \valid(p+(0..n-1));
     terminates \true;
     ensures \forall integer i; 0<=i<n ==> p[i]==0;
     assigns p[0..n-1];
     exits \false;
   @/
  void VM_zeroOut(unsigned int n, \ghost unsigned int * p);
*/

// Logic predicates/functions...

/*@
  predicate VM_UnchangedGlobalExceptState{K,L}  =
       (\forall integer i; 0<=i<\at(VM_NP,K) ==>
        \at(VM_sc[i],K)==\at(VM_sc[i],L))
    && (\forall integer i; 0<=i<\at(VM_NP,K) ==>
        \at(VM_rc[i],K)==\at(VM_rc[i],L));

  predicate VM_UnchangedGlobal{K,L} = VM_UnchangedGlobalExceptState{K,L} &&
    \at(VM_state,K) == \at(VM_state,L);
 */

#define VM_UnchangedLocal(K,L) \
  \at(VM_currentRegion,K)==\at(VM_currentRegion,L) &&   \
    \at(VM_regionCount,K)==\at(VM_regionCount,L) &&     \
    \at(VM_nregions,K)==\at(VM_nregions,L) &&           \
    \at(VM_level,K)==\at(VM_level,L)

#define VM_Unchanged(K,L)                               \
  VM_UnchangedGlobal{K,L} && VM_UnchangedLocal{K,L}

// Havoc functions needed in non-ghost code for receives...

//void VM_havoc_double(double * buf) /*@ ghost(int count) */;

#define VM_HAVOC_CONTRACT            \
  requires \valid(buf+(0..count-1)); \
  terminates \true;                  \
  assigns buf[0..count-1];           \
  exits \false;

/*@ VM_HAVOC_CONTRACT */
void VM_havoc_int(int * buf, int count);

/*@ VM_HAVOC_CONTRACT */
void VM_havoc_float(float * buf, int count);

/*@ VM_HAVOC_CONTRACT */
void VM_havoc_double(double * buf, int count);

// add havocs for more basic C types here as needed.

#define VM_havoc(p, n)                          \
  _Generic((p),                                 \
    int*: VM_havoc_int(p, n),                   \
    float*: VM_havoc_float(p, n),               \
    double*: VM_havoc_double(p, n)              \
  )

// the following is a trick to (seemingly) get a ghost value into the
// non-ghost state...
/*@ terminates \true;
    ensures \result == x;
    assigns \nothing;
    exits \false;
*/
int VM_anyint(void) /*@ ghost (int x) */;

// begin a region.  call in ghost code only:
#define VM_Region_begin(fid,rid)                                \
  /@ assert VM_EB1: VM_currentRegion == VM_NO_REGION; @/        \
  /@ assert VM_EB2: rid == VM_EXPECTED(fid,VM_regionCount); @/  \
  VM_currentRegion = rid;                                       \
  VM_regionCount++;

// end a region.  call in ghost code only:
#define VM_Region_end(fid,rid)                                          \
  /@ assert VM_EE1: VM_currentRegion == rid; @/                         \
  /@ assert VM_EE2: \forall integer _i;                                 \
     0<=_i<VM_NP ==> VM_sc[_i]==VM_NUMMSG(fid,rid,VM_pid,_i); @/        \
  /@ assert VM_EE3: \forall integer _i;                                 \
     0<=_i<VM_NP ==> VM_rc[_i]==VM_NUMMSG(fid,rid,_i,VM_pid); @/        \
  VM_zeroOut(VM_NP, VM_sc);                                             \
  VM_zeroOut(VM_NP, VM_rc);                                             \
  VM_level = 0;                                                         \
  VM_currentRegion = VM_NO_REGION;

// receive a message.   call in ghost code only:
#define VM_recv(fid,rid,buf,count,dt,src,tag,comm,stat)                 \
  /@ assert VM_R1: VM_state == VM_Active; @/                            \
  /@ assert VM_R2: rid==VM_currentRegion; @/                            \
  if (src != MPI_PROC_NULL) {                                           \
    unsigned _rc = VM_rc[src];                                          \
    int _mcount = VM_MCOUNT(fid,rid,src,VM_pid,_rc);                    \
    ulong _rlevel = VM_SLEVEL(fid,rid,src,VM_pid,_rc);                  \
    /@ assert VM_R3:  _mcount <= count; @/                              \
    /@ assert VM_R4:  \valid((buf)+(0 .. _mcount-1)); @/                \
    /@ assert VM_R5:  dt == VM_MDTYPE(fid,rid,src,VM_pid,_rc); @/       \
    /@ assert VM_R6:  0<=src<VM_NP; @/                                  \
    /@ assert VM_R7:  tag==MPI_ANY_TAG ||                               \
                      tag==VM_MSGTAG(fid,rid,src,VM_pid,_rc); @/        \
    /@ assert VM_R8:  comm==MPI_COMM_WORLD; @/                          \
    /@ assert VM_R9:  stat==MPI_STATUS_IGNORE; @/                       \
    /@ assert VM_R10: _rlevel > VM_level; @/                            \
    VM_level = _rlevel;                                                 \
    /@ admit VM_MSGINV(fid,rid,src,VM_pid,_rc,buf,_mcount,dt); @/       \
    VM_rc[src]++;                                                       \
  }   

// the following goes in ghost code only:
#define VM_send(fid,rid,buf,count,dt,dest,tag,comm)                     \
  /@ assert VM_S1: VM_state == VM_Active; @/                            \
  /@ assert VM_S2: rid==VM_currentRegion; @/                            \
  if (dest != MPI_PROC_NULL) {                                          \
    unsigned _sc = VM_sc[dest];                                         \
    int _mcount = VM_MCOUNT(fid,rid,src,VM_pid,_sc);                    \
    ulong _slevel = VM_SLEVEL(fid,rid,VM_pid,dest,_sc);                 \
    /@ assert VM_S3: count == _mcount; @/                               \
    /@ assert VM_S4: \valid_read((buf)+(0 .. _mcount-1)); @/            \
    /@ assert VM_S5: dt==VM_MDTYPE(fid,rid,VM_pid,dest,_sc); @/         \
    /@ assert VM_S6: 0<=dest<VM_NP; @/                                  \
    /@ assert VM_S7: tag==VM_MSGTAG(fid,rid,VM_pid,dest,_sc); @/        \
    /@ assert VM_S8: comm==MPI_COMM_WORLD; @/                           \
    /@ assert VM_S9: _slevel > VM_level; @/                             \
    VM_level = _slevel;                                                 \
    /@ assert VM_S10:                                                   \
       VM_MSGINV(fid,rid,VM_pid,dest,_sc,buf,_mcount,dt); @/            \
    VM_sc[dest]++;                                                      \
  }

#define VM_sendrecv(fid,rid,sbuf,scount,sdt,dest,stag,                  \
                    rbuf,rcount,rdt,src,rtag,comm,stat)                 \
  {                                                                     \
    ulong _old_level = VM_level;                                        \
    /@ assert VM_SR1: VM_state == VM_Active; @/                         \
    /@ assert VM_SR2: rid==VM_currentRegion; @/                         \
    /@ assert VM_SR3: comm==MPI_COMM_WORLD; @/                          \
    if (dest != MPI_PROC_NULL) {                                        \
      unsigned _sc = VM_sc[dest];                                       \
      int _mcount = VM_MCOUNT(fid,rid,src,VM_pid,_sc);                  \
      ulong _slevel = VM_SLEVEL(fid,rid,VM_pid,dest,_sc);               \
      /@ assert VM_SR4: scount == _mcount; @/                           \
      /@ assert VM_SR5: \valid_read((sbuf)+(0 .. _mcount-1)); @/        \
      /@ assert VM_SR6: sdt==VM_MDTYPE(fid,rid,VM_pid,dest,_sc); @/     \
      /@ assert VM_SR7: 0<=dest<VM_NP; @/                               \
      /@ assert VM_SR8: stag==VM_MSGTAG(fid,rid,VM_pid,dest,_sc); @/    \
      ulong _new_level = VM_SLEVEL(fid,rid,VM_pid,dest,_sc);            \
      /@ assert VM_SR9: _new_level > _old_level; @/                     \
      VM_level = _new_level;                                            \
      /@ assert VM_SR10:                                                \
         VM_MSGINV(fid,rid,VM_pid,dest,_sc,sbuf,_mcount,sdt); @/        \
      VM_sc[dest]++;                                                    \
    }                                                                   \
    if (src != MPI_PROC_NULL) {                                         \
      unsigned _rc = VM_rc[src];                                        \
      int _mcount = VM_MCOUNT(fid,rid,src,VM_pid,_rc);                  \
      ulong _rlevel = VM_SLEVEL(fid,rid,src,VM_pid,_rc);                \
      /@ assert VM_SR11:  _mcount <= rcount; @/                         \
      /@ assert VM_SR12:  \valid((rbuf)+(0 .. _mcount-1)); @/           \
      /@ assert VM_SR13:  rdt == VM_MDTYPE(fid,rid,src,VM_pid,_rc); @/  \
      /@ assert VM_SR14:  0<=src<VM_NP; @/                              \
      /@ assert VM_SR15:  rtag==MPI_ANY_TAG ||                          \
           rtag==VM_MSGTAG(fid,rid,src,VM_pid,_rc[src]); @/             \
      ulong _new_level = VM_SLEVEL(fid,rid,src,VM_pid,_rc);             \
      /@ assert VM_SR16: _new_level > _old_level; @/                    \
      if (_new_level > VM_level) VM_level = _new_level;                 \
      /@ assert VM_SR17: stat==MPI_STATUS_IGNORE; @/                    \
      /@ admit VM_MSGINV(fid,rid,src,VM_pid,_rc,rbuf,_mcount,rdt); @/   \
      VM_rc[src]++;                                                     \
    }                                                                   \
  }

/*@ requires VM_state == VM_Active;
    requires comm == MPI_COMM_WORLD;
    requires \valid(size);
    terminates \true;
    ensures \result==0 && *size == VM_NP;
    assigns *size;
    exits \false;
 */
int MPI_Comm_size(MPI_Comm comm, int * size);

/*@ requires VM_state == VM_Active;
    requires comm == MPI_COMM_WORLD;
    requires \valid(rank);
    terminates \true;
    ensures \result==0 && *rank == VM_pid;
    assigns *rank;
    exits \false;
 */
int MPI_Comm_rank(MPI_Comm comm, int * rank);

/*@ requires VM_state == VM_Uninitialized;
    requires argc == NULL && argv == NULL;
    terminates \true;
    ensures \result==0 && VM_state == VM_Active;
    assigns VM_state;
    exits \false;
*/
int MPI_Init(int *argc, char ***argv);

/*@ requires VM_state == VM_Active;
    terminates \true;
    ensures \result==0 && VM_state == VM_Finalized;
    assigns VM_state;
    exits \false;
*/
int MPI_Finalize(void);

#endif
\end{lstlisting}

\end{document}